\documentclass[smallextended]{svjour3}

\usepackage{a4wide}
\usepackage{amsmath}
\usepackage{amsfonts}
\usepackage{amssymb}
\usepackage{mathrsfs} 
\usepackage[dvipsnames]{xcolor}
\usepackage{graphicx}
\usepackage{color}
\usepackage{array}

\usepackage{tikz,tikz-qtree}

\usepackage{pifont}
\newcommand{\cmark}{\ding{51}}%
\newcommand{\xmark}{\ding{55}}%

\usepackage{hyperref}

\newtheorem{thm}{Theorem}
\numberwithin{thm}{section} 
\newtheorem{rem}{Remark}
\newtheorem{lem}[thm]{Lemma}
\newtheorem{mydef}{Definition}
\numberwithin{mydef}{section} 

\newcommand{\beqn}{\begin{eqnarray}\begin{aligned}}
\newcommand{\eqn}{\end{aligned}\end{eqnarray}}
\newcommand{\fra}[2]{\textstyle{\frac{#1}{#2}}}

\def\ov{\overline}
\def\1{\ov{1}}
\def\2{\ov{2}}
\def\3{\ov{3}}

\definecolor{mygreen}{rgb}{0.2656,0.5039,0.2148}
\definecolor{myred}{rgb}{0.75,0,0.25}

\begin{document}
\title{Developing a statistically powerful measure for quartet tree inference using phylogenetic identities and Markov invariants}


%
\author{Jeremy G. Sumner\and Amelia Taylor \and  Barbara~R.~Holland \and Peter D. Jarvis}
\institute{J G Sumner, B R Holland, P D Jarvis \\\email{Jeremy.Sumner@utas.edu.au}
\at School of Physical Sciences, University of Tasmania, Australia
\and
A Taylor  %
\at Oregon State University, USA
}


\maketitle

\vspace{1em}

\begin{abstract}
Recently there has been renewed interest in phylogenetic inference methods based on phylogenetic invariants, alongside the related Markov invariants. 
Broadly speaking, both these approaches give rise to polynomial functions of sequence site patterns that, in expectation value, either vanish  for particular evolutionary trees (in the case of phylogenetic invariants) or have well understood transformation properties (in the case of Markov invariants). 
While both approaches have been valued for their intrinsic mathematical interest, it is not clear how they relate to each other, and to what extent they can be used as practical tools for inference of phylogenetic trees.

In this paper, by focusing on the special case of binary sequence data and quartets of taxa, we are able to view these two different polynomial-based approaches within a common framework. 
To motivate the discussion, we present three desirable statistical properties that we argue any invariant-based phylogenetic method should satisfy: (1) sensible behaviour under reordering of input sequences; (2) stability as the taxa evolve independently according to a Markov process; and (3) explicit dependence on the assumption of a continuous-time process.
Motivated by these statistical properties, we develop and explore several new phylogenetic inference methods. 
In particular, we develop a statistically bias-corrected version of the Markov invariants approach which satisfies all three properties. 
We also extend previous work by showing that the phylogenetic invariants can be implemented in such a way as to satisfy property (3). 
A simulation study shows that, in comparison to other methods, our new proposed approach based on bias-corrected  Markov invariants is extremely powerful for phylogenetic inference.

The binary case is of particular theoretical interest as --– in this case only --- the Markov invariants can be expressed as linear combinations of the phylogenetic invariants. 
A wider implication of this is that, for models with more than two states --- for example DNA sequence alignments with four-state models --- we find that methods which rely on phylogenetic invariants are incapable of satisfying all three of the stated statistical properties. 
This is because in these cases the relevant Markov invariants belong to a class of polynomials independent from the phylogenetic invariants.

\end{abstract}

\section{Introduction and motivation}

In the late 1980s, Cavender and Felsenstein \cite{cavender1987} and Lake \cite{lake1987} introduced the idea of \emph{phylogenetic invariants}; a class of polynomials useful in the study of phylogenetic trees.   
In subsequent years, these polynomials have proven useful for studying analytical questions of identifiability \cite{allman2003} and for identifying local maximum likelihood optima \cite{chor2000}.
However, beginning with the earliest simulation studies~\cite{hillis1994}, there has been doubt as to the statistical effectiveness of phylogenetic invariants for inference of phylogenetic trees from data sets.

With their paper~\cite{allman2008}, Allman and Rhodes renewed interest in phylogenetic invariants. They took the point of view of algebraic geometry to give a comprehensive description of these polynomials and lay out several open questions (some of which have subsequently been solved \cite{allman2014semialgebraic,bates2011toward,friedland2013tensors,friedland2012proof}). 
Concurrently, Sumner and coauthors \cite{sumner2008} suggested an alternative perspective on algebraic methods as applied to phylogenetics. 
From this perspective, group representation theory (symmetries and transformations) takes center stage, leading to the study of a different set of polynomials of special interest, the \emph{Markov invariants}. 
In contrast to phylogenetic invariants, the definition of Markov invariants is detached from the notion of a phylogenetic tree; rather they are the polynomial invariants for the matrix group induced by the action of Markov matrices. 
As such, the application of Markov invariants to the context of phylogenetics comes only after consideration of the specific tree structures underlying phylogenetic models. 
In this vein, Sumner and Jarvis \cite{sumner2009} showed how leaf permutation symmetries on a quartet tree, for example, can be used to bring Markov invariants into phylogenetics proper; effectively showing there are phylogenetic invariants lurking within the ring of Markov invariants applicable to this case.  
Recently, both perspectives have been applied to inferring phylogenetic trees~\cite{casanellas2010,holland2012}, with further promising results appearing in \cite{fernandez2015}. 

Most likely due to the disjointed historical development of these polynomial functions, there is some confusion, already clear in the paragraph above, regarding the use of ``invariant'' as applied to both phylogenetic and Markov invariants.   
In the literature~\cite{allman2003,allman2008,cavender1987} ``phylogenetic invariant'' is used to refer to any polynomial which vanishes on all distributions arising from a subset of phylogenetic tree topologies (understood as leaf-labelled trees). 
If the subset is proper, the phylogenetic invariant is referred to as ``tree informative''.   
We, however, prefer to use ``invariant'' in the more mathematically traditional sense to mean invariant under an invertible transformation (c.f. classical invariant theory \cite{olver2003}). 
We argue that in the phylogenetic context, the relevant transformations are adjustments of model parameters and leaf permutations of trees.   
To avoid confusion, we follow Draisma and Kuttler~\cite{draisma2008} and refer to any polynomial which is useful for identifying tree topology as a \emph{phylogenetic identity}.  
In contrast, we say a polynomial is a \emph{Markov invariant} \cite{sumner2008} if the polynomial itself (rather than its particular value on subsets of  distributions) is invariant under adjustment of model parameters on a phylogenetic tree (the precise meaning of this distinction is made clear in Section~\ref{sec:background}).
Formally, these polynomials are  invariant  under a specific action of a group of invertible transformations (at least ``relatively'', that is, they may attract a transformation constant). 
Clearly distinguishing Markov invariants from phylogenetic identities is crucial in what follows. 



Given that phylogenetic identities arise solely from \emph{algebraic} conditions on phylogenetic probability distributions, we argue it is also essential to consider the \emph{statistical} structure of inference methods constructed using these polynomials more carefully than has previously appeared in the literature. 
Toward this end, we provide a comprehensive discussion, including both analytical and statistical arguments and a comparison of the algebraic geometry and representation theory perspectives, of using the phylogenetic identities for the inference of phylogenetic trees. 
To simplify the discussion, we focus on the most elementary case: quartet trees with a binary state space.  
We argue that the representation theoretic point of view and the ideas underlying Markov invariants provide significant  guidance as to how to construct statistically powerful methods of phylogenetic inference.  

Binary state spaces have long been of theoretical interest in the study of phylogenetic methods as the mathematical properties of two-state models are often more tractable, and yet the results are still illuminating about general phylogenetic principles. 
We also note that recently there has been increased interest in binary data from an applied point of view due to the widespread availability of bi-allelic single nucleotide polymorphism (SNP) datasets derived from modern genome-wide sequencing technologies \cite{davey2011genome,lemmon2013high}.

Our discussion is unified through two notions of symmetry that naturally arise in phylogenetics. 
In Section~\ref{sec:background} we develop these and refer to them as ``leaf symmetries'' and the ``Markov action''.
In Section~\ref{sec:quartetInference} we argue that any inference method that seeks to infer tree topology alone (as is typical of phylogenetic identity methods) should respect both of these symmetries.  
We show that respect for leaf permutation symmetry is something that can (and should) be imposed upon any tree inference method based on phylogenetic identities. 
Additionally, demanding the method respect the Markov action symmetry leads directly to the definition of Markov invariants, with our main example constructed in Section~\ref{sec:squangles}. 
An ideal situation arises in the quartet case: we show that imposing the leaf permutation symmetry upon the Markov invariants identifies a specific subset of phylogenetic identities, which in turn leads to a \emph{unique} choice of identities to apply to quartet tree inference. 

In Section~\ref{sec:edgeIdentities} we discuss the properties of the edge identities; especially in relation to the three statistical properties given in Section~\ref{sec:quartetInference}.
We provide a detailed examination of the behaviour of the edge identities under leaf permutation symmetries and, as for the squangles, derive semi-algebraic constraints for their behaviour under the assumption of a continuous-time Markov chain.

Along with our theoretical arguments for considering polynomials which respect these two symmetries, we also use these symmetries to develop a statistical decision rule for tree inference (via residual sums of squares). 
In Section~\ref{sec:simulations}, we provide simulation studies which illustrate both the practical importance of these ideas and that the naive application of phylogenetic identities (like that in~\cite{cavender1987}) can be statistically biased and not nearly as powerful as our approach motivated by the symmetries inherent to the problem.

In Section ~\ref{sec:generalize}, we conclude with a discussion of how these ideas apply directly to models with more than binary states, with specific results presented for the four state (DNA) case. 
In particular, we find that it is only in the binary case that the Markov invariants (squangles) lie in the same space of polynomials as the phylogenetic identities (edge identities). 
Thus in the case of models with greater than two states, the attractive transformation properties of the Markov invariants become a missed opportunity if one restricts attention to edge identities (as is advocated in \cite{casanellas2010}). 
This result is derived using representation theoretical techniques (particularly group characters \cite{jarvis2014adventures}) for which full derivations are provided in the Appendix (Online Resource 4).

\section{Background}
\label{sec:background}

In phylogenetic inference, the topology of the evolutionary tree is difficult to determine correctly and is often the unknown parameter which is the most biologically important. 
It is well known that it is enough to correctly identify all the quartet trees corresponding to all subsets of four taxa in order to determine the overall phylogenetic tree.
Thus correct identification of a single quartet topology remains a point of considerable mathematical interest, and is the focus of the work presented here.  

\begin{rem}
Throughout this paper we will exclusively consider phylogenetic quartet inference methods that, given aligned sequence data on four taxa as input, solely return confidence in each of the three possible quartet tree topologies. 
For methods (such as maximum likelihood) that usually also return estimates of evolutionary divergence times or other model parameters, we will consider the topology to be the only output.
\end{rem}

\subsection{Taxon permutations and leaf symmetries}

\label{sec:symmetries}

When discussing four general taxa, we label them $A,B,C,D$; and when we want to discuss a fixed order on the taxa we label them $1, 2, 3, 4$.  
This gives us a natural way to talk both about the three different quartet trees and equivalent trees using the common split notation.  
In this notation, the three distinct quartet trees are $T_1 \!=\! 12|34$, $T_2 \!=\! 13|24$ and $T_3 \!=\! 14|23$, where formally $ij|kl\equiv \{\{i,j\},\{k,l\}\}$ is a bipartition of the set $\{1,2,3,4\}$. 
Each quartet has symmetries under leaf permutations which are captured by the equalities $12|34\!=\!21|34 \!=\!34|12\ldots$ etc.  
 These different representations of the same quartet are of practical importance if one considers the application of a phylogenetic method (usually via some computer software) on the taxon set $\{A,B,C,D\}$ with output one of the quartets $T_1,T_2,T_3$. 
For instance, if the list of taxa in the ordering $A,B,C,D$ leads to $T_1$ we would expect the alternative input ordering $A,C,B,D$ to return the quartet $T_2$ (since $B$ now corresponds to 3, and $C$
to 2), and the alternative input order $D,C,A,B$ to also return $T_1$ via the correspondence $12|34\!=\!43|12$.  


Such changes in taxon ordering can be understood as the symmetric group $\mathfrak{S}_4$ permuting the four taxa in the natural way, thereby inducing permutations of the three possible quartet trees. 
For example, the taxon permutation $(13)\in \mathfrak{S}_4$ fixes $T_2$ and interchanges $T_1\leftrightarrow T_3$.
From the perspective of phylogenetic quartet inference, we account for this redundancy by considering the subgroup of $\mathfrak{S}_4$ that fixes a given quartet.  
For example, $T_1$ is invariant under the action of the subgroup of $\mathfrak{S}_4$ consisting of the permutations which we refer to as the  \emph{stabilizer of} $T_1$:
\[
\text{Stab}(T_1)=\{e,(12),(34),(12)(34),(13)(24),(14)(23),(1324),(1423)\}.
\]
It is an easy exercise to write down the stabilizer subgroups for $T_2$ and $T_3$.

To understand the importance of these observations, consider the ``black box'' view of a phylogenetic quartet method, where the black box (in the form of a computer program\footnote{It is sometimes important to distinguish between a method (as theoretically conceived) and its implementation in software (for example, ambiguities often arise in quartet methods in regard to random tie breaking). Throughout this article we will assume the two match up perfectly without further comment.}) takes an \emph{ordered} set of taxon sequences $A,B,C,D$, and returns one of the three possible quartets $T_1,T_2$ or $T_3$. 
To say that the method ``respects'' the permutation symmetries explained above is to demand that the method behaves in the appropriate way given a permutation of the input sequences  such as $B,A,C,D$ corresponding to $(12)$, or $C,D,A,B$ corresponding to the permutation $(13)(24)$.  
We ensure that the phylogenetic methods we develop in this paper respect these quartet tree leaf permutation symmetries.

\subsection{Tensors and group actions}
\label{sec:groupaction}

The data we consider are frequency arrays $F=\left(f_{ijkl}\right)$ arising from an alignment of four binary $\{0,1\}$ sequences, where $f_{ijkl}$ is the number of times we observe the pattern of states $i, j, k, l$ for sequence $1, 2, 3, 4$, respectively.

We model this data by assuming $F$ arises under multinomial (independent) sampling from a distribution $P=(p_{ijkl})$ which itself is constructed from  a binary Markov chain on a quartet tree, where $p_{ijkl}$ is the probability of observing the binary states $i,j,k,l\in \{0,1\}$ at the leaves $1,2,3,4$ of the tree, respectively. 
In the next section we discuss the construction of such $P$ in detail; for the moment we wish to consider the generic structural properties of $P$ irrespective of whether $P$ arises as a probability distribution on a tree or not.

Considering $P=(p_{ijkl})$ as a $2\times 2\times 2\times 2$ array of numbers, and taking $\{e_1=\left[\begin{smallmatrix}1\\0\end{smallmatrix}\right],e_2=\left[\begin{smallmatrix}0\\1\end{smallmatrix}\right]\}$ as a basis for $\mathbb{C}^2$, allows us to treat $P$ more formally as belonging to the $2^4=16$ dimensional tensor product space 
\[
U:= \mathbb{C}^2\otimes
\mathbb{C}^2\otimes\mathbb{C}^2\otimes\mathbb{C}^2=\left\{\sum_{i,j,k,l\in\{0,1\}}p_{ijkl}e_i\otimes e_j\otimes e_k\otimes e_l: p_{ijkl}\in \mathbb{C}\right\}.
\]
Of course, $P$ has all real and non-negative components so $P$ actually belongs to a stochastic subset of this space.
However, algebraically it is convenient to work over the complex numbers in what follows.  
When speaking abstractly we refer to a general member of $U$ as a \emph{tensor}, and when we want to emphasize that the components in the array should be considered as probabilities, we will refer to it as a \emph{distribution}.   

The taxon permutations discussed in the previous section act naturally on tensors $P\in U$ via permutation of the indices of $p_{ijkl}$.
To be concrete, suppose $\sigma \in \mathfrak{S}_4$ is a permutation, then we have the action $P\mapsto \sigma P$ defined via the coordinate transformation $p_{i_1i_2i_3i_4}\mapsto p_{i_{\sigma(1)}i_{\sigma(2)}i_{\sigma(3)}i_{\sigma(4)}}$.

Another key mathematical feature of working with tensor product spaces, essential to our derivations, is the natural action of the general linear group $\text{GL}(2)$ on each factor of the tensor product space, described as follows. 
Recall that $\text{GL}(2)$ is the group of $2\times 2$ invertible matrices with entries taken from $\mathbb{C}$, that is
\[
\text{GL}(2)=
\left\{
A=
\left[
\begin{matrix}
a_{11} & a_{12} \\
a_{21} & a_{22}
\end{matrix}
\right]
:
a_{11},a_{12},a_{21},a_{22}\in \mathbb{C},\det(A)\neq 0
\right\}.
\]
Recall also that $\text{GL}(2)$ acts on column vectors $v=[v_1,v_2]^T\in \mathbb{C}^2$ via $v\mapsto Av$ or, equivalently, in component form: $v_i\mapsto \sum_{i'\in \{0,1\}} a_{ii'}v_{i'}$.
This action extends to $U$ by taking four matrices $A,B,C,D\in \text{GL}(2)$ and defining an analogous rule for tensor component transformations: 
\[
p_{ijkl}\mapsto \sum_{ i',j',k',l'\in \{0,1\}} a_{ii'}b_{jj'}c_{kk'}d_{ll'}p_{i'j'k'l'}.
\]
This provides an action of the direct product group $\times^4 \text{GL}(2)\equiv \text{GL}(2)\!\times\! \text{GL}(2)\!\times\! \text{GL}(2) \!\times\!\text{GL}(2)$ expressed in tensor form as the mapping $P\mapsto A\otimes B\otimes C\otimes D\cdot P$.



In what follows, we consider the actions of both $\mathfrak{S}_4$ and $\times^4\text{GL}(2)$ on tensors $P\in U$.
For the former with $\sigma\in\mathfrak{S}_4$ we will generically write $P\mapsto \sigma\cdot P$, and for the latter with $g=A\otimes B\otimes C\otimes D\in \times^4\text{GL}(2)$ we will generically write $P\mapsto g\cdot P$.
Although in this notation there is ambiguity between which group action we are applying, we will resolve this in all cases by providing the necessary context.


%

\subsection{Tree tensors, clipped tensors}
\label{sec:markovgroup}

We will say that $M$ is a ($2\!\times\! 2$) \emph{Markov matrix} if
\[
M=
\left[
\begin{matrix}
1-a_{21} & a_{12} \\
a_{21} & 1-a_{12}
\end{matrix}
\right],
\]
where $0 \leq a_{12},a_{21} \leq 1$ are the probabilities of state changes $0\rightarrow 1$ and $1\rightarrow 0$, respectively.
We consider the \emph{rooted} version of the quartet tree $T_1$ obtained by placing an additional  vertex (the ``root'') on the internal edge of $T_1$. 
We label each edge of the tree by the subset of leaves descendant to the edge.
Let $\pi=[\pi_1,\pi_2]^T$ be a probability distribution (that is, $\pi_i> 0$ and $\pi_1+\pi_2=1$), and let $M_{e}=(m_{ij}^{(e)})$ be a collection of Markov matrices indexed by the edges $e\in T_1$\footnote{To avoid unimportant technicalities, we will assume ``generic'' parameter settings throughout this article. In particular, we assume that all Markov matrices are non-singular and $\pi_i\neq 0$ for $i=1,2$.}.
We set
\[
p^{(1)}_{ijkl}=\sum_{x,y,r\in \{0,1\}} m^{(1)}_{ix}m^{(2)}_{jx}m^{(3)}_{ky}m^{(4)}_{ly}m^{(12)}_{xr}m^{(34)}_{yr}\pi_r.
\]

Under this construction, the tensor $P_1=(p^{(1)}_{ijkl})$ corresponds to the standard construction of a probability distribution arising from the Markov process on $T_1$ (as described in textbooks such as~\cite{felsenstein2004}).
Additionally, a well-known result (a generalization of Felsenstein's ``pulley-principle'' \cite{felsenstein1981}) shows it is possible to adjust the free parameters in this expression such that we can move the root of $T_1$ to anywhere we please, whilst fixing the distribution $P_1$.
Motivated by this: 
\begin{mydef}
We say that a tensor $P_1$ is a \textbf{tree tensor corresponding to the quartet} $T_1=12|34$ if $P_1$ arises under the construction just given, for any choice of Markov matrices, root distribution, and root placement.
Similarly, we say that $P_2$ and $P_3$ are tree tensors corresponding to the quartets $T_2=13|24$ and $T_3=14|23$ if they arise in the analogous way on the remaining two quartets.
\end{mydef}

We now connect this construction to our description of the natural action of $\times^4\text{GL}(2)$ on $U$ described in the previous section.
We do this by defining, for any fixed tree tensor $P_i$, the \emph{clipped tensor} $\widetilde{P}_i$, which is obtained by setting
each Markov matrix on the leaf edges of the quartet to be equal to the identity matrix.
In this way, generically we have (for example):
\beqn
\label{eq:clipped}
\widetilde{p}^{(1)}_{ijkl}=
\left\{
\begin{matrix}
\sum_{r\in \{0,1\}} m^{(12)}_{ir}m^{(34)}_{kr}\pi_r,\text{ if }i=j\text{ and }k=l;\\
0, \text{ otherwise}.
\end{matrix}
\right.
\eqn

From the definitions given in the previous section, we can now write
\[
P_1=M_1\otimes M_2\otimes M_3\otimes M_4\cdot \widetilde{P}_1,
\]
and consider $P_1$ as arising from the clipped tensor $\widetilde{P}$ under the action of $\times^4\text{GL}(2)$ (provided we make the additional assumption that each of the Markov matrices $M_1,M_2,M_3, M_4$ is invertible and hence belongs to $\text{GL}(2)$).
This motivates:

\begin{mydef}
The \textbf{Markov group} $\mathcal{M}_2$ is the set of matrices:
\[
\mathcal{M}_2=
\left\{
M=
\left[
\begin{matrix}
1-a_{21} & a_{12} \\
a_{21} & 1-a_{12}
\end{matrix}
\right]:
a_{12},a_{21}\in \mathbb{C}, \det(M)\neq 0 
\right\}.
\]
\end{mydef}
Notice we have removed the stochastic constraints on the matrix entries so that $\mathcal{M}_2$ is a proper subgroup of $\text{GL}(2)$ (as is easy to verify).

While this perspective excludes tree tensors constructed using non-invertible Markov matrices, this is not a serious objection since, from a modelling perspective, we prefer to take the point of view of
continuous-time Markov chains where all relevant Markov matrices are invertible (since they occur as matrix exponentials). 
In any case, within the set of Markov matrices the subset with zero determinant is of measure zero and hence we may assume that any Markov matrix occurring in practice (in a sufficiently random way) will indeed belong to $\mathcal{M}_2$.
Thus we may consider tree tensors $P_i$ as arising under the action of $\times^4 \mathcal{M}_2$, as a subgroup of $\times^4 \text{GL}(2)$, on clipped tensors $\widetilde{P}_i$.

\subsection{Markov action}\label{sec:markovAction}


Conceptually, we can extend the notion of the action of $\times^4\mathcal{M}_2$ on clipped tensors $\widetilde{P}_i$ to an action on \emph{all}  tensors in $U$.
Of particular importance is the following: if $P_1\in U$ is a tree tensor and $M_1,M_2,M_3,M_4\in \mathcal{M}_2$ are Markov matrices, we can interpret the action 
\[
P_1\mapsto M_1\otimes M_2\otimes M_3\otimes M_4\cdot P_1
\] 
as corresponding to lengthening the leaves of the phylogenetic tree. 
Of course this interpretation works for any tensor $P\in U$ (whether $P$ is a tree tensor or otherwise).
\begin{mydef}
\label{def:markovaction}
The \textbf{Markov action} is the group action of $\times^4\mathcal{M}_2$ on $U$ obtained by restricting each copy of $\text{GL}(2)$ in $\times^4\text{GL}(2)$ to the Markov group $\mathcal{M}_2$.
\end{mydef}


Importantly, this action encodes the conditional independence of Markov evolution across lineages; and, if $P$ happens to be a tree tensor, \emph{this action preserves the underlying tree topology}.
In other words, the Markov action provides a symmetry on the set of quartet tree tensors.
Connecting this with our previously discussed black box view, where a quartet method is assumed to estimate tree topology only, we see that the Markov action is essentially a nuisance parameter that ideally the method should be insensitive to.

\subsection{Markov invariants}

With the Markov action in hand we can now formally define the polynomials that are our main interest in this paper.
This class of polynomials was first defined and explored in \cite{sumner2008}.

\begin{mydef}
\label{def:markovinvariants}
Take $q(P)$ to be a multivariate polynomial function on the indeterminates $P=(p_{ijkl})$.
We say that $q$ is a \textbf{Markov invariant} if $q$ transforms as a one-dimensional representation under the Markov action.
\end{mydef}
In the language of classical invariant theory, this  is equivalent to saying $q$ is a ``relative invariant'' under the Markov action so, for all $P\in U$ and all $g=M_1\otimes M_2\otimes M_3\otimes M_4\in \times^4\mathcal{M}_2$:
\[
q(g\cdot P)=q(M_1\otimes M_2\otimes M_3\otimes M_4\cdot P)=\lambda_gq(P),
\]
where $\lambda_g\in \mathbb{C}$ satisfies, for all $g,g'\in\times^4 \mathcal{M}_2$, the multiplicative property: $\lambda_{gg'}=\lambda_g\lambda_{g'}$.
In the language of group representation theory, this means that $\lambda_g$ provides a one-dimensional representation of $\times^4 \mathcal{M}_2$.
In the examples we discuss, $\lambda_g$ is simply a power of the determinant $\det(g)$ (from which the multiplicative property follows easily).

As alluded to in the previous section, our interest in Markov invariants is motivated by the desire to control the behaviour, under the Markov action, of a quartet phylogenetic method founded on the evaluation of a set of polynomials.
The Markov invariants represent the optimal case where we have complete understanding of what is happening under the Markov action.
As we will see, the situation is quite different for the classically constructed phylogenetic identities.


\subsection{Flattenings, minors, and edge identities}
\label{sec:flattenings}

Here we derive the so-called ``edge identities''.
In their most general form, these are phylogenetic identities for phylogenetic trees, which can be used to detect the presence or absence of a particular edge in the phylogenetic tree.
These identities were first derived using the general concepts of tensor flattenings and associated rank conditions in \cite{allman2008}.
Here we specialize to the case of binary states and quartet trees and take an approach which focuses on the role of the Markov action.

\begin{mydef}
Suppose $P=(p_{i_1i_2i_3i_4})\in U$ is a generic tensor and suppose $\alpha \beta|\gamma\delta$ is a bipartition of $\{1,2,3,4\}$.
The \textbf{flattening} of $P$ corresponding to the bipartition $\alpha \beta|\gamma\delta$ is the $2^2\times 2^2$ matrix containing the entries $p_{i_1i_2i_3i_4}$ with rows indexed by $ i_\alpha i_\beta=00,01,10,11 $ and columns indexed by $ i_\gamma i_\delta=00,01,10,11$.  
\end{mydef}
Up to row and column permutations, there are only three distinct flattenings of a tensor $P\in U$, each corresponding to one of the possible quartet trees $T_1,T_2$ or $T_3$. 
Concretely, we denote the ``$12|34$'' flattening of $P$ as the $4\times 4$ matrix $\text{Flat}_1(P)$ with entries
\[
{\text{Flat}_1(P)}_{i_1i_2,i_3i_4}=p_{i_1i_2i_3i_4}.
\]
Similarly we define the ``$13|24$'' and  ``$14|23$'' flattenings as the $4\times 4$ matrices $\text{Flat}_2(P)$ and $\text{Flat}_3(P)$ with entries
\[
{\text{Flat}_2(P)}_{i_1i_3,i_2i_4}= p_{i_1i_2i_3i_4},\qquad
{\text{Flat}_3(P)}_{i_1i_4,i_2i_3}= p_{i_1i_2i_3i_4},
\]
respectively.

The action of $\times^4\text{GL}(2)$ discussed in Section~\ref{sec:groupaction}, $P\rightarrow A\otimes B\otimes C\otimes D \cdot P$, can be shown to be expressed on the $12|34$ flattening as
\beqn
\label{eq:flataction}
\text{Flat}_1(P)\rightarrow (A\otimes B)\cdot \text{Flat}_1(P) \cdot (C\otimes D)^T, 
\eqn
where ${}^T$ indicates matrix transpose\footnote{At a formal level (not strictly required here), the reader should note that since we are working over the complex field, the flattening should be defined so in place of the matrix transpose  in (\ref{eq:flataction}) we have the \emph{conjugate} transpose operation.}.


Using the flattenings, it is not too hard to derive some phylogenetic identities for quartet trees.
Consider a clipped tensor $\widetilde{P_1}$ from the quartet tree $T_1$ and its flattening
\beqn
\label{eq:rank2}
\text{Flat}_1(\widetilde{P}_1)=
\left[
\begin{matrix}
x & 0 & 0 & y \\
0 & 0 & 0 & 0\\
0 & 0 & 0 & 0\\
z & 0 & 0 & w 
\end{matrix}
\right],
\eqn
where the $x,y,z,w$ label the non-zero probabilities given in (\ref{eq:clipped}).
From (\ref{eq:flataction}) we see that
\[
\text{Flat}_1(P_1)=M_1\otimes M_2 \cdot \text{Flat}_1(\widetilde{P}_1)\cdot (M_3\otimes M_4)^T,
\]
and hence, assuming each $M_i\in\mathcal{M}_2$ is non-singular, we conclude that $\text{rank}\left(\text{Flat}_1(P_1)\right)\leq 2$.
Therefore the 16 cubic polynomials obtained by taking 3-minors of this flattened matrix form a set of phylogenetic identities for the quartet $12|34$.
We refer to these minors as \emph{edge identities}.
(We will see in Section~\ref{sec:squangles} that these minors are actually tree-informative since they \emph{do not} vanish on the other two quartets, at least generically.)

These observations generalize to:
\begin{thm}\cite{allman2008}
In each case $i\!=\!1,2,3$; the 16 polynomial functions in the indeterminates $P_i=(p^{(i)}_{i_1i_2i_3i_4})$  produced by taking the cubic $3$-minors of the flattening $\text{Flat}_i(P_i)$ form phylogenetic identities for probability distributions $P_i$ arising from the quartet tree $T_i$. 
\end{thm}

An attractive feature of this process of taking flattenings and minors is that the construction can be generalized to phylogenetic tensors with any number of taxa, and Markov chains with arbitrary state spaces (beyond the binary case discussed here).
This observation was first presented in \cite{allman2008} and generalized to a wider class of models in \cite{draisma2008} and \cite{casanellas2010}.


\section{Quartet inference measures}~\label{sec:quartetInference}

We now describe some desirable properties of any quartet method which returns tree topology only.
We suppose the pattern frequency array $F=(f_{ijkl})\in U$ for four taxa arose as $N$ independent samples from some fixed distribution $P\in U$.
(In particular one may like to consider the case where $P=P_i$ arose on the tree $T_i$, but this is not necessary for the discussion in this section.)
We interpret $N$ as sequence length of the alignment, and denote this situation as $F\sim \text{MultiNom}(P,N)$, noting this implies $F$ has (componentwise) expectation value $E[F]=NP$.

\begin{mydef}\label{def:measure}
A triple $\Delta(F)=(R_1,R_2,R_3)$ is called a \textbf{quartet inference measure} (or simply a \textbf{measure}) for $F$ if  each $R_1$,$R_2$,$R_3$ is a (statistically interpretable)  confidence in the respective statements $F \sim \text{MultiNom}(P_1,N)$, $F \sim \text{MultiNom}(P_2,N)$, $F \sim \text{MultiNom}(P_3,N)$, for some $P_1$,$P_2$,$P_3$ arising in turn from the quartets $T_1$,$T_2$,$T_3$.
\end{mydef}

Later, we set each $R_i$ equal to a residual sum of squares under the quartet hypothesis $T_i$, but for the moment we assume, without loss of generality, that $\Delta$ is designed so that small values of $R_i$ correspond to greater confidence in quartet $T_i$.
Given this, we assume the quartet inference measure $\Delta$ ranks the statistical confidence in the three quartet trees $T_1$, $T_2$ and $T_3$ using the relative ordering of  $R_1$, $R_2$ and $R_3$.

Considering quartet inference measures $\Delta$ in the abstract sense, in Table~\ref{tab:properties} we describe three theoretical statistical properties a measure may, or may not, satisfy.
On the practical side, in the simulation study (Section~\ref{sec:simulations}), we apply several specific examples of quartet measures $\Delta$ constructed from polynomial functions (both phylogenetic identities and Markov invariants) on the tensor product space $U$.
The results of the simulations clearly establish the importance of each of the properties given in Table~\ref{tab:properties}.

Presently, we illustrate the three properties by showing:

\begin{table}[ht]
\caption{Proposed desirable statistical properties of quartet inference measures.}
\begin{tabular}{l}

\fcolorbox{black}{lightgray}{
\begin{minipage}{.9\textwidth}
\smallskip
\noindent
\textbf{Property I}

\smallskip
\emph{The quartet inference measure $\Delta(F)$ should satisfy an explicit transformation rule under taxon permutations.}
\medskip

\noindent
In detail, this means that if $\Delta(F)=(R_1,R_2,R_3)$ and we permute the taxa so $F'\!=\!\sigma\cdot  F$ with $\sigma \in \mathfrak{S}_4$, then 
\[
\Delta(F')=(R_{\sigma'(1)},R_{\sigma'(2)},R_{\sigma'(3)}),
\]
where $\sigma\mapsto \sigma'\in \mathfrak{S}_3$ is the homomorphism induced by considering taxon permutations as leaf permutations on quartet trees (as discussed in Section~\ref{sec:symmetries}).
\end{minipage}
}

\\
\\

\fcolorbox{black}{lightgray}{
\begin{minipage}{.9\textwidth}
\smallskip
\noindent
\textbf{Property II (strong version)}

\smallskip 
\emph{In expectation value, the quartet inference measure $\Delta$ should satisfy an explicit transformation rule under the Markov action.}
\medskip

\noindent
For instance, if $F\sim \text{MultiNom}(P,N)$,  $g\in\times^4\mathcal{M}_2$, and $F'\sim \text{MultiNom}(g\cdot P,N)$ there exists a scalar $\lambda_g$ such that:
\[
E[\Delta(F')]= \lambda_g E[\Delta(F)],
\]
where $\lambda_g$ satisfies the \emph{multiplicative} group homomorphism property $\lambda_{g}\lambda_{g'}=\lambda_{gg'}$ for all $g,g'\times^4\mathcal{M}_2$.

\smallskip
Alternatively, $\lambda_g$ may satisfy the \emph{additive} group homomorphism property $\lambda_{g}+\lambda_{g'}\!=\!\lambda_{gg'}$ so that $E[\Delta(F')]=  E[\Delta(F)]+(\lambda_g,\lambda_g,\lambda_g)$.
\end{minipage}
}

\\
\\

\fcolorbox{black}{lightgray}{
\begin{minipage}{.9\textwidth}
\smallskip
\noindent
\textbf{Property II (weak version)} 

\smallskip
\emph{In expectation value and in the limit of infinite sequence length, the quartet inference measure $\Delta(F)$ should satisfy an explicit transformation rule under the Markov action.}
\smallskip

\noindent
For instance, as in the strong version but with equality true in the limit of infinite sequence length $N$ (assuming the multiplicative property for $\lambda_g$):
\[
\lim_{{N\rightarrow \infty} }E[\Delta(F')]=\lim_{{N\rightarrow \infty} }\lambda_g E[\Delta(F)].
\]
\end{minipage}
}

\\
\\

\fcolorbox{black}{lightgray}{
\begin{minipage}{.9\textwidth}
\smallskip
\noindent
\textbf{Property III} 

\smallskip
\emph{The quartet inference measure $\Delta(F)$ should be explicitly dependent on the assumption of a continuous-time process.}

\noindent

\end{minipage}
}
\smallskip

\end{tabular}
\label{tab:properties}
\end{table}

\begin{thm}
\label{thm:nj}
The neighbor-joining algorithm \cite{saitou1987} together with an additive estimator of pairwise distance consistent with a fixed Markov model  provides a quartet inference measure satisfying Property I, Property II (strong), and Property III. 
\end{thm}
\noindent
Note: Supposing  the pairwise distance estimator between taxa $i$ and $j$ input to neighbor-joining is denoted as $d_{ij}$.
By ``additive'' and ``consistent with a given Markov model'' we mean the following:
\begin{enumerate}
\item A specific continuous-time Markov model on quartet trees is fixed;
\item The associated Markov matrices produce a matrix group \cite{sumner2011} so the ``Markov action'' on the leaves is well defined (as in Definition~\ref{def:markovaction}) ;
\item The expectation value $E[d_{ij}]$ is equal to the sum of the branch lengths on the path from leaf $i$ to $j$ . 
\end{enumerate}
Examples of Markov models where these conditions can be achieved include the binary-symmetric and Jukes-Cantor models, together with their unbiased pairwise distance estimators (see, for example, \cite{felsenstein2004}).
In the following we give an outline of a proof.
\begin{proof}
For quartets, the neighbor-joining algorithm returns the quartet corresponding to the minimum of the three-tuple $\Delta=(R_1,R_2,R_3):=(d_{12}+d_{34},d_{13}+d_{24},d_{14}+d_{23})$.
Under this definition, it is clear that $\Delta$ satisfies Property I, as required.

Further, if each $d_{ij}$ is additive and consistent with a Markov model on the tree (as described above), then under the Markov action it follows that $E[R_i]\rightarrow E[R_i]+\lambda_g$, where $\lambda_g:=t_1'+t_2'+t_3'+t_4'$ and each $t_i'$ is the extended branch length on leaf $i$ of the quartet. 
Setting $\Delta(F)=(R_1,R_2,R_3)$  we have, under the leaf action: $E[\Delta(F')]=E[\Delta(F)]+(\lambda_g,\lambda_g,\lambda_g)$, and $\lambda_g+\lambda_{g'}=\lambda_{gg'}$ (since the branch lengths are additive under further extension of the leaves).
This establishes that, under these conditions, neighbor-joining satisfies Property II (strong), as required.

Finally, Property III is built into our assumption on the pairwise distance estimator.
(For example, the Jukes-Cantor distance estimator for DNA sequences will fail to return a finite answer when the proportion of sites that differ in a pairwise sequence alignment is greater than $0.75$; this is a structural feature resulting from a continuous-time assumption.)
\hfill $\square$
\end{proof}

Thinking in a continuous-time formulation of Markov chains, in general one would expect that such $\lambda_g$ would have some monotonicity property with respect to time so that, up to a fixed amount of statistical noise, our ability to discriminate quartets using a measure $\Delta$ satisfying Property II decreases in time.
This is indeed the case for the example of neighbor-joining just given, and more generally corresponds to the biological fact that the ability to detect homology between extant taxa (that is, the ``phylogenetic signal'') degrades as the divergence of common ancestry is pushed further backwards in time.
In our application of Markov invariants, we will see that this is also the case where $\lambda_g$ is multiplicative and $\lambda_g\sim e^{-\gamma t}$, with $\gamma>0$.

Previous work has discussed applying Property I \cite{eriksson2008,sumner2009,rusinko2011} in the context of phylogenetic identities.
To our knowledge, Property II has never been explicitly discussed before.
We will however show in Section~\ref{sec:squangles} that Property II (weak) is implicit in the quartet method based on Markov invariants presented in \cite{holland2012}.

We are convinced that these properties of a quartet measure $\Delta$ are natural given that the purpose of $\Delta$ is to deliver confidence in the choice of quartet from observed data.
We will explain how the Markov invariants are ideally tailored to the task of constructing quartet measures that satisfy Property II in its strong version.
As we will see, this is contingent upon the construction of unbiased estimators of the Markov invariants; a problem we solve completely in the binary quartet case, but is otherwise open (see Section~\ref{sec:generalize}).

The next two sections contain the derivations of Markov invariants and the related discussion of Properties I and II.  


\section{The squangles}
\label{sec:squangles}
As previously noted in Section~\ref{sec:markovAction}, whether a phylogenetic pattern distribution $F$ arises as a sample from a specific quartet $T_i$ depends only on the \emph{internal} structure of the tree, not on the lengths or model parameters on the leaf edges. 
This motivates Definition~\ref{def:markovinvariants} of Markov invariants, which for historical reasons in the quartet case on four-state, DNA models, we call ``squangles'' (\textbf{s}tochastic \textbf{q}uartet t\textbf{angle}, see \cite{sumner2008}). 
We work with an analogous construction in the binary case and, when needed, refer to these polynomials as ``binary squangles" or, whenever there is no risk of ambiguity, simply as ``squangles''.  

In this section, we first derive the (binary) squangles, then use them to build a quartet measure $\Delta$ which satisfies Properties I, II (weak), and III.  
We then consider issues of statistical bias to build a second measure that satisfies Properties I, II (strong), and III.  

\subsection{Construction}
\label{subsec:construct}
To motivate and construct the squangles, we use an alternative basis for $\mathbb{C}^2$.
Our choice of basis is motivated by the simple observation that a linear change of coordinates on the probability vectors $[p_0,p_1]^T$ makes probability conservation, $p_0+p_1\!=\!1$ an explicitly conserved quantity under the action of Markov matrices $\mathcal{M}_2$.

To this end, we use the orthogonal similarity transformation 
$h=\frac{1}{\sqrt{2}}\left[
\begin{smallmatrix}
1 & 1 \\ 
-1 & 1
 \end{smallmatrix}\right]$ with inverse $h^{-1}\!=\!h^T$, so that $2\times 2$ Markov matrices
$M=\left[\begin{smallmatrix}1-a & b \\ a &
    1-b\end{smallmatrix}\right]$ are transformed to 
$
M'=h^TMh=\left[
\begin{smallmatrix}
\lambda & v \\
0 & 1
\end{smallmatrix}
\right]
$,
where $\lambda\!=\!1\!-\!a\!-\!b$ and $v\!=\!b\!-\!a$, and the second row explicitly manifests probability conservation. 
In what is to come, we will have additional recourse to consider only parameters that arise under a continuous-time formulation of a Markov chain, so that $M=e^{Qt}$ for some $2\times 2$ ``rate'' (zero-column sum) matrix $Q$. 
In this case we have the constraints $0\leq a,b< \frac{1}{2}$ which, in particular, implies $0< \lambda\leq 1$.

Let $P\in U$ be a distribution with components $p_{ijkl}$. 
Following the notation set out in Section~\ref{sec:flattenings} we let $\text{Flat}_1(P)$ be the $12|34$ flattening of $P$, which under the Markov action transforms as
\[
\text{Flat}_1(P)\rightarrow (M_1\otimes M_2)\cdot \text{Flat}_1(P)\cdot (M_3\otimes M_4)^T.
\]
In the alternative basis we have the $4\times 4$ form
\[
M'_1\otimes M'_2=
\left[
\begin{matrix}
\lambda_1\lambda_2 & \lambda_1v_2 & v_1\lambda_2 & v_1v_2 \\
0 & \lambda_1 & 0 & v_1 \\
0 & 0 & \lambda_2 & v_2 \\
0 & 0 & 0 & 1
\end{matrix}
\right],
\]
and a similar expression for $M'_3\otimes M'_4$.
Commensurately, we let $\text{Flat}_1'(P)$ denote the $12|34$ flattening in the alternate basis:
\[
\text{Flat}_1'(P):=\left(h^T\otimes h^T\right)\cdot \text{Flat}_1(P)\cdot \left(h\otimes h\right).
\]

This formulation allows us to identify the bottom right $3\times 3$ sub-matrix $\widehat{\text{Flat}}'_1(P)$ of $\text{Flat}'_1(P)$ as providing an invariant subspace for the Markov action, that is 
\[
\widehat{\text{Flat}}'_1(P)\rightarrow \widehat{(M'_1\otimes M'_2)}\cdot \widehat{\text{Flat}}'_1(P)\cdot \widehat{(M'_3\otimes M'_4)}^T,
\]
where\footnote{This observation admits a significant generalization --- presented in \cite{sumner2016dimensional} --- to any number of taxa and any number of states $k$.} 
\[
\widehat{M'_1\otimes M'_2}:=
\left[
\begin{matrix}
\lambda_1 & 0 & v_1 \\
0 & \lambda_2 & v_2 \\
0 & 0 & 1
\end{matrix}
\right],
\] 
and similarly for $\widehat{M'_3\otimes M'_4}$.

Further, this construction leads to a cubic Markov invariant using nothing more than the multiplicative property of the determinant: 
\beqn
\label{eq:sqprop}
\det(\widehat{\text{Flat}}'_1(P))\rightarrow \det&\left(\widehat{(M'_1\otimes M'_2)}\cdot \widehat{\text{Flat}}'_1(P)\cdot (\widehat{M'_3\otimes M'_4})^T\right)\\
&\hspace{8em}=\det(\widehat{M'_1\otimes M'_2})\det(\widehat{\text{Flat}}_1'(P))\det(\widehat{M'_3\otimes M'_4})\\
&\hspace{8em}=\lambda_1\lambda_2\lambda_3\lambda_4\det(\widehat{\text{Flat}}'_1(P)).
\eqn
As a polynomial on $U$, we set $q_1(P):=\det(\widehat{\text{Flat}}'_1(P))$ and $q_1$ is our first example of a Markov invariant on the space of tensors $U$ since, for all $P\in U$ and $g\!=\!M_1\otimes M_2\otimes M_3\otimes M_4\in \times^4\mathcal{M}_2$, we have:
\[
q_1(g\cdot P)=\det(g)q_1(P),
\]
where $\det(g)\!=\!\det(M_1)\det(M_2)\det(M_3)\det(M_4)\equiv \lambda_1\lambda_2\lambda_3\lambda_4$.
Thus:
\begin{thm}
\label{thm:thesquangle}
The polynomial $q_1$ defined as $q_1(P):=\det(\widehat{\text{Flat}}'_1(P))$ is a Markov invariant accompanied by the one-dimensional representation of $\times^4\mathcal{M}_2$ given by $\lambda_g=\det(g)$ for all $g\in \times^4\mathcal{M}_2$.
\end{thm}
For the reasons explained at the start of this section, we refer to $q_1$ as the ``squangle''.

The reader should note that the squangle $q_1$ is defined via (and depends absolutely upon) both the $12|34$ flattening and our particular choice of basis for $\mathbb{C}^2$.  
On the other hand, $q_1(P)$ is perfectly well defined for \emph{all} tensors $P\in U$, and occurs as a homogeneous, cubic polynomial in the indeterminates $p_{i_1i_2i_3i_4}$ with 96 terms (the explicit polynomial form is provided in Online Resource 1).

It is also important to note that (\ref{eq:sqprop}) is valid only under the action of $2\times 2$ \emph{Markov} matrices and certainly fails for more general $2\times 2$ matrices in $\text{GL}(2)$. 
Thus the squangles are very much tailored for the probabilistic setting of Markov chains. 

Having constructed $q_1$ we now  evaluate $q_1$ specifically on a tensor $P_1$ arising from the quartet tree $T_1$ with the goal of producing a quartet inference measure $\Delta$. 
As observed in Section~\ref{sec:markovgroup}, if $P_1$ arises from a quartet we can certainly write $P_1=M_1\otimes M_2\otimes M_3\otimes M_4 \cdot \widetilde{P}_1$,
where $\widetilde{P}_1$ is the so-called clipped tensor. 
In particular, in the original probability basis, this tensor has components $\widetilde{p}^{(1)}_{ijkl}= 0$ whenever $i\!\neq\! j$ or $k\!\neq\! l$.
 
  
We also saw in (\ref{eq:rank2}) that, under the $12|34$ flattening, $\text{Flat}_1(\widetilde{P}_1)$ generically has rank at most 2.
Hence, working in the alternative basis, $\widehat{\text{Flat}}'_1(\widetilde{P}_1)$ also has rank at most 2, and being a cubic minor we obtain
\[
q_1(\widetilde{P}_1)=\det(\widehat{\text{Flat}}'_1(\widetilde{P}_1))=0\quad\implies\quad  
q_1(P_1)=\lambda_1\lambda_2\lambda_3\lambda_4q_1(\widetilde{P}_1)=0,
\]
for all tensors $P_1$ arising on the quartet tree $12|34$ under any choices of parameters.
Thus: 
\begin{thm}
\label{thm:q1identity}
The Markov invariant $q_1$ is a phylogenetic identity for the quartet $12|34$. 
\end{thm}


On the other hand if we suppose a distribution $P_2$ arises from the quartet tree $13|24$ we can write $P_2=M_1\otimes M_2\otimes M_3\otimes M_4 \cdot \widetilde{P}_2$,
where, considered as the $12|34$ flattening in the original basis, we have generically:
\[
\text{Flat}_1(\widetilde{P}_2)=
\left[
\begin{matrix}
x & 0 & 0 & 0 \\
0 & y & 0 & 0\\
0 & 0 & z & 0\\
0 & 0 & 0 & w 
\end{matrix}
\right].
\]
Transforming to the alternative basis and evaluating $q_1(\widetilde{P}_2)$  we now find
\[
q_1(\widetilde{P}_2)=\frac{1}{4}(wyz+xyz+wxy+wxz).
\]
Since $\widetilde{P}_2$ is a distribution we have $x,y,z,w>0$ and hence
$
q_1(\widetilde{P}_2)>0
$.
Since $q_1$ is a Markov invariant, we have $q_1(P_2)=\lambda_1\lambda_2\lambda_3\lambda_4q_1(\widetilde{P}_2)$, and we conclude $q_1(P_2)>0$ for all choices of parameters such that $\widetilde{P}_2$  corresponds to a probability distribution on the quartet $13|24$ under continuous-time formulation of a Markov chain where $0<\lambda_i=\det(M_i)=e^{\text{tr}(Q_i t)}\leq 1$.


Finally if we suppose $P_3$ arises from $T_3$ we get 
\[
P_3=M_1\otimes M_2\otimes M_3\otimes M_4 \cdot \widetilde{P}_3,
\]
and again under the $12|34$ flattening in the original basis, we have
\[
\text{Flat}_1(\widetilde{P}_3)=
\left[
\begin{matrix}
x & 0 & 0 & 0 \\
0 & 0 & y & 0\\
0 & z & 0 & 0\\
0 & 0 & 0 & w 
\end{matrix}
\right],
\]
which follows simply from the structural property of the components of $P_3$ in the original basis: $\widetilde{p}_{ijkl}\neq 0$ if and only if $i\!=\!l$ and $j\!=\!k$.
Transforming to the alternative basis and evaluating $q_1(\widetilde{P}_3)$  we now find
\[
q_1(\widetilde{P}_3)=-\frac{1}{4}(wyz+xyz+wxy+wxz).
\]
Since $q_1$ is a Markov invariant, we have $q_1(P_3)=\lambda_1\lambda_2\lambda_3\lambda_4q_1(\widetilde{P}_3)$, and we conclude $q_1(P_3)<0$ for sensible choices of parameters, that is, parameters such that $\widetilde{P}_3$ really does correspond to a probability distribution and, on the leaf edges, $0<\det(M_i)\leq 1$.

We can of course define two additional squangles $q_2,q_3$ using the other two choices of tensor flattenings $13|24$ and $14|23$.
This can be achieved using an analogous argument to the one we gave for $q_1$ but it is simpler at this stage to utilize the natural action of $\mathfrak{S}_4$ on tensors $P\in U$ to define:
\[
q_2(P):=-q_1((23)\cdot P),\qquad
q_3(P):=-q_1((24)\cdot P);
\]
where the choice of signs is  chosen for reasons of elegance that will become apparent.

Clearly $q_2$ and $q_3$ also form Markov invariants since:
\beqn
q_2(M_1\otimes M_2 \otimes M_3 \otimes M_4 \cdot P)&=-q_1((23)\cdot M_1\otimes M_2 \otimes M_3 \otimes M_4 \cdot P)\\
&=-q_1(M_1\otimes M_3 \otimes M_2 \otimes M_4 \cdot (23)\cdot P)\\
&=-\lambda_1\lambda_3\lambda_2\lambda_4q_1((23)\cdot P)=\lambda_1\lambda_2\lambda_3\lambda_4q_2(P),\nonumber
\eqn
with a similar derivation for $q_3$.

\subsection{Signs for the squangles} 
A critical part of our construction of a useful measure for tree inference relies on understanding the expected values of the polynomials, and particularly, their expected signs.  Thus, we use the invariance property established in the last subsection to infer positivity conditions for $q_2$ and $q_3$ on the three possible quartets as follows (note we have already established these conditions for $q_1$ as part of our development of the last section). 

Suppose $P_2$ is a tensor arising from the quartet $13|24$.
As before we can write $P_2=M_1\otimes M_2\otimes M_3\otimes M_4\cdot \widetilde{P}_2$.
Now taking $\widetilde{P}_1:=(23)\cdot \widetilde{P}_2$, it is clear that $\widetilde{P}_1$ is a clipped tensor taken from the quartet $12|34$.
Thus
\[
q_2(\widetilde{P}_2)=-q_1((23)\cdot\widetilde{P}_2)=-q_1(\widetilde{P}_1)=0,
\]
since we concluded above that $q_1(\widetilde{P}_1)=0$ for all tensors from the quartet $12|34$.
Conversely, choosing any clipped tensor $\widetilde{P}_1$ from $12|34$ and defining $\widetilde{P}_2:=(23)\cdot\widetilde{P}_1$, we have:
\[
q_2(\widetilde{P}_1)=-q_1((23)\cdot\widetilde{P}_1)=-q_1(\widetilde{P}_2)<0.
\]
Continuing in this fashion we infer the signs of the evaluations of the squangles on tensors from the three possible quartets.

Before we summarize this information however, we note the squangles form a vector space (linear combination of these invariant functions is again an invariant function), and explicit computation shows that this vector space only has dimension two, that is, there is a linear dependence between the polynomials $q_1,q_2,q_3$.
In fact, this dependency is exhibited by
\[
q_1+q_2+q_3=0,
\]
as a polynomial identity.
Thus only two of the squangles are needed to span the vector space of these invariant functions.
In the Appendix (Online Resource 4) we show that this linear dependence follows directly from a representation theoretic argument using group characters.
Given this linear dependence, for reasons of symmetry it makes
sense to consider the pair $\{q_2,q_3\}$ as a basis for the squangles when the quartet $12|34$ is under consideration, the pair $\{q_1,q_3\}$ as a basis when the quartet $13|24$ is under consideration, and the pair $\{q_1,q_2\}$ as a basis when the quartet $14|23$ is under consideration.

Putting the information found so far together, we find expected values for the squangles $q_1$, $q_2$, and $q_3$ when evaluated on the three possible quartets as given in Table~\ref{tab:expect}.
We use this table of expectation values to design an optimal quartet inference measure $\Delta$.

\begin{thm}
\label{thm:sqexp}
Given a probability tensor $P_i\in U$ arising from the quartet $T_i$, when evaluated on a frequency array $F\sim \text{MultiNom}(P_i,N)$, the Markov invariants $\{q_1,q_2,q_3\}$ have the signed expectation values given in Table~\ref{tab:expect}.
\end{thm}
\begin{proof} 
Given the above observations regarding the signs of the squangles when evaluated on the three possible quartets, to complete the proof, we need only confirm, for $i\!=\!1,2,3$, the expectation values $E[q_i(F)]=N(N-1)(N-2)q_i(P)$ for all probability tensors $P$ and $F\sim \text{MultiNom}(P,N)$.

Under the multinomial distribution, we have $E[F]=NP$ which is simply the vector version of $E[f_{ijkl}]=Np_{ijkl}$.
In general, the situation for higher monomial powers in the $f_{ijkl}$ is not so straightforward.
However, the explicit polynomial form given in Online Resource 1 reveals that each monomial term in $q_1$ is square free (and hence the same result follows for the squangles $q_2$ and $q_3$).
Considering a square free cubic monomial $f_{i_1j_1k_1l_1}f_{i_2j_2k_2l_2}f_{i_3j_4k_4l_4}$, one finds, using the moment generating function for the multinomial distribution (see (\ref{eq:mogen}) and surrounding discussion below):
\beqn
E[f_{i_1j_1k_1l_1}f_{i_2j_2k_2l_2}f_{i_3j_3k_3l_3}]=N(N-1)(N-2)p_{i_1j_1k_1l_1}p_{i_2j_2k_2l_2}p_{i_3j_3k_3l_3}.\nonumber
\eqn
We then apply linearity of expectation value to conclude $E[q_i(F)]=N(N-1)(N-2)q_i(P)$ for $i=1,2,3$.
Defining $u:=E[q_1(F)]$, $v:=E[q_2(F)]$, and $w:=E[q_3(F)]$ completes the proof.
\hfill $\square$
\end{proof}

\begin{table}[h]
\caption{Expectation values of the three squangles $q_1,q_2,q_3$ when evaluated on a tree tensor $P_i$ corresponding to quartet $T_i$. 
Under a continuous-time assumption, the expectation values $u\!\equiv\!u(P_1),v\!\equiv\!v(P_2),w\!\equiv\!w(P_3)$ satisfy the constraints $u,v,w\geq  0$, but are otherwise unknown and depend upon the specific model parameters.}
\centering
\begin{tabular}{c|ccc}
& $T_1$ & $T_2$ & $T_3$\\
\hline
$E[q_1(F)]$ & $\phantom{-}0$ & $\phantom{-}v$ & $-w$ \\
$E[q_2(F)]$ & $-u$ & $\phantom{-}0$ & $\phantom{-}w$ \\
$E[q_3(F)]$ & $\phantom{-}u$ & $-v$ & $\phantom{-}0$ \\
\end{tabular}
\label{tab:expect}
\end{table}

Before using this information to derive a quartet inference measure, we first need to consider the behaviour of the squangles under taxon permutations.

\subsection{Taxon permutations for the squangles}
\label{sec:taxonperms}
From the definition of the flattenings and the action of $\mathfrak{S}_4$ on $U$ it follows that
\[
\text{Flat}_1((12)\cdot P)=K\text{Flat}_1(P),\qquad \text{Flat}_1((13)(24)\cdot P)=\text{Flat}_1(P)^T,
\]
where $K$ is the permutation matrix
\[
K=
\left[
\begin{matrix}
1 & 0 & 0 & 0\\
0 & 0 & 1 & 0\\
0 & 1 & 0 & 0\\
0 & 0 & 0 & 1\\
\end{matrix}
\right].
\]
Further, since the permutations $(12)$ and $(13)(24)$ generate the stabilizer $\text{Stab}(T_1)$, we see that the action of the eight permutations in $\text{Stab}(T_1)$ comes from compositions of these basic two:
\beqn
\begin{matrix}
\text{Flat}_1(P), & K\text{Flat}_1(P),& \text{Flat}_1(P)K,& K\text{Flat}_1(P)K,\\\text{Flat}_1(P)^T,&K\text{Flat}_1(P)^T,&\text{Flat}_1(P)^TK,&K\text{Flat}_1(P)^TK.\nonumber
\end{matrix}
\eqn

Transforming this result into the alternative basis, it is straightforward to show:
\[
\text{Flat}'_1((12)\cdot P)=K\text{Flat}'_1(P),\qquad \text{Flat}'_1((13)(24)\cdot P)=\text{Flat}'_1(P)^T,
\]
where, in the first result, we have used $(h^T\otimes h^T)\cdot K\cdot (h\otimes h)=K$.
From this we see that 
\[
q_1((12)\cdot P)=\det(\widehat{\text{Flat}}'_1((12)\cdot P))=\det(K)\det(\widehat{\text{Flat}}'_1( P))=-\det(\widehat{\text{Flat}}'_1( P))=-q_1(P), 
\]
and similarly $q_1((13)(24)\cdot P)=q_1(P)$.
Thus the squangle $q_1$ spans a one-dimensional subspace under the action of the stabilizer $\text{Stab}(T_1)$. 
In particular, $q_1$ transforms as the $\texttt{sgn}$ representation of $\mathfrak{S}_4$ restricted to the stabilizer:

\begin{thm}
\label{thm:sqleaf}
The squangle $q_1$ transforms  as $\texttt{sgn}$ under the action of the stabilizer $\text{Stab}(T_1)$ defined by $q_1(P)\mapsto q_1(\sigma\cdot P)=\text{sgn}(\sigma)q_1(P)$, for all $\sigma\in \text{Stab}(T_1)$ and $P\in U$.
\end{thm}

Following the approach of \cite{sumner2009} for the DNA (four-state) case, this result provides an alternative route to establishing that the squangle $q_1$ is a phylogenetic identity for quartet $T_1$ (Theorem~\ref{thm:q1identity}) using the notion of the clipped tensor, as follows.
If $P_1$ is a tree tensor corresponding to quartet $T_1$, then it is clear that the clipped tensor $\widetilde{P}_1$ is fixed under the (odd) permutation $(12)\in \text{Stab}(T_1) $, that is $(12)\cdot \widetilde{P}_1=\widetilde{P}_1$.
Hence,
\beqn
q_1(P_1)&=\lambda_1\lambda_2\lambda_3\lambda_4q_1(\widetilde{P}_1)\nonumber=\lambda_1\lambda_2\lambda_3\lambda_4q_1((12)\cdot \widetilde{P}_1)=\lambda_1\lambda_2\lambda_3\lambda_4\texttt{sgn}((12))q_1( \widetilde{P}_1)=-q_1(P_1),
\eqn
and we conclude that $q_1(P_1)=0$ for all tree tensors $P_1$ corresponding to quartet $T_1$.

Now considering the following calculation:
\beqn
q_2((12)\cdot P)&=-q_1((23)(12)\cdot P)\\
&=-q_1((132)\cdot P)\\
&=-q_1((12)(13)\cdot P)=q_1((13)\cdot P)=q_1((13)(24) (24)\cdot P)=q_1((24)\cdot P)=-q_3(P).\nonumber
\eqn
Where we have used the definition of $q_2$ in the first equality, Theorem~\ref{thm:sqleaf} in the fourth and fifth equality, and the definition of $q_3$ in the final equality.
A similar calculation shows:
\[
q_2((13)(24)\cdot P)=q_2(P).
\]
From this, we conclude:
\begin{thm}
\label{thm:squanglesleaf}
The squangles $q_2$ and $q_3$ transform under the action of the stabilizer $\text{Stab}(T_1)$ as a signed permutation representation. 
Specifically, for all $\sigma\in \text{Stab}(T_1)$ and $P\in U$:
\[
q_2(\sigma P)=
\left\{
\begin{matrix}
q_2(P),\text{ if }\texttt{sgn}(\sigma)=1;\\
-q_3(P), \text{ if }\texttt{sgn}(\sigma)=-1.
\end{matrix}
\right.
\]
\end{thm}

In the next section, we will apply these results to construct quartet measures which explicitly satisfy Property I.

\subsection{The measure and residual sum of squares}
\label{sec:ResSums}

We are now ready to discuss specific examples of quartet inference measures $\Delta$.  
Given the expectation values given in Table~\ref{tab:expect}, we may construct a naive measure using the squangles as $\Delta(F) = (|q_1(F)|^\ell, |q_2(F)|^\ell, |q_3(F)|^\ell)$ for some integer $\ell>0$. 
We note that Theorems~\ref{thm:thesquangle}, \ref{thm:sqleaf} and \ref{thm:squanglesleaf}  suggest that this measure may satisfy Properties I, II (possibly in the strong form), and III.
However, we need to consider the statistical situation carefully to establish this formally and, as we will see, this motivates us to consider a more sophisticated  measure.

Firstly, we consider the \emph{signed} expectations given in Table~\ref{tab:expect} and develop a residual sum of squares measure --- analogous to that used in~\cite{holland2012} --- that takes these signs into account.  
For purposes of self-containment, we revisit the derivation and then, after a consideration of statistical bias correction under multinomial sampling, modify the quartet inference measure to produce one that satisfies Properties I, II (strong), and III, as described in Section~\ref{sec:quartetInference}.     

Suppose we are interested in the hypothesis that the array of observed pattern frequencies occurs as a multinomial sample $F\sim \text{MultiNom}(P_1,N)$ with $P_1$ arising on quartet $T_1$ under some fixed set of parameters.
Evaluating the squangles on the array $F$, we see that our best estimate of the parameter $u\geq 0$ is given by 
\beqn
\hat{u}=
\left\{
\begin{matrix}
\frac{1}{2}(q_3(F)-q_2(F)),\text{ if } q_3(F)> q_2(F);\\
0,\text{ otherwise}.
\end{matrix}
\right.\nonumber
\eqn
If $\hat{u}>0$ then the residual sum of squares is 
\beqn
(q_2(F)+\hat{u})^2+(q_3(F)-\hat{u})^2=\frac{1}{2}(q_2(F)+q_3(F))^2=\frac{1}{2}q_1^2(F),\nonumber
\eqn
since $q_1+q_2+q_3=0$.
On the other hand, if $\hat{u}= 0$, the residual sum of squares is $q^2_2(F)+q_3^2(F)$.
The residuals for the other two quartet hypotheses can be obtained similarly and all are presented in Table~\ref{tab:residuals}.
These results exactly correspond to those given in \cite{holland2012} for the DNA squangles case. 
Presently, we take these ideas further by considering issues of statistical bias to find an inference measure which, in expectation value, satisfies Properties I, II (strong), and III.

\begin{table}[h]
\caption{Residual sums of squares for each quartet hypothesis and possible ordering of squangle values.
For each ordering,  residuals which are plausibly minimal are highlighted. 
The key point to observe is that the first three orderings perfectly match the ordering of expectation values, and in this case there is only one plausible minimal residual and hence we conclude that the corresponding quartet is most likely.
In the other three cases, the orderings do not match the ordering of expectation values and there are two competing quartets.
}
\centering
\begin{tabular}{c|ccc}
 & $\text{RSS}_1$ & $\text{RSS}_2$ & $\text{RSS}_3$\\
 \hline
$q_2\leq q_1 \leq q_3$ & $\mathbf{\fra{1}{2}q_1^2(F)}$ & $q_1^2(F)+q_3^2(F)$ &  $q_1^2(F)+q_2^2(F)$\\
$q_3\leq q_2 \leq q_1$ & $q_2^2(F)+q_3^2(F)$ & $\mathbf{\fra{1}{2}q_2^2(F)}$  &  $q_1^2(F)+q_2^2(F)$\\
$q_1\leq q_3 \leq q_2$ & $q_2^2(F)+q_3^2(F)$ & $q_1^2(F)+q_3^2(F)$  &  $\mathbf{\fra{1}{2}q_3^2(F)}$\\
$q_1\leq q_2 \leq q_3$ & $\mathbf{\fra{1}{2}q_1^2(F)}$ & $q_1^2(F)+q_3^2(F)$ & $\mathbf{\fra{1}{2}q_3^2(F)}$ \\
$q_3\leq q_1 \leq q_2$ & $q_2^2(F)+q_3^2(F)$ & $\mathbf{\fra{1}{2}q_2^2(F)}$ & $\mathbf{\fra{1}{2}q_3^2(F)}$ \\
$q_2\leq q_3 \leq q_1$ & $\mathbf{\fra{1}{2}q_1^2(F)}$ & $\mathbf{\fra{1}{2}q_2^2(F)}$ & $q_1^2(F)+q_2^2(F)$
\end{tabular}
\label{tab:residuals}
\end{table}

To motivate the discussion, assume $F\sim \text{MultiNom}(P_1,N)$ where $P_1$ is a distribution arising from quartet $T_1$ and suppose $q_2(F)\leq q_1(F)\leq q_3(F)$.
Then, under the least squares approach, we have the residual sum of squares for each quartet hypothesis $T_i$:
\[
\Delta(F) = (\text{RSS}_1,\text{RSS}_2,\text{RSS}_3)=\left(\fra{1}{2}q_1^2(F),q_{1}^2(F)+q_{3}^2(F),q_{1}^2(F)+q_{2}^2(F)\right).
\]

To ensure Property II (strong) we need the expected value of $\Delta$ for the situation $F\sim \text{MultiNom}(P_1,N)$ to be proportional to the expected value of $\Delta$ for $F'\sim \text{MultiNom}(g\cdot P_1,N)$ with $g\in \times^4\mathcal{M}_2$.
However this is not true as the situation stands since
\[
E[\Delta]=
E[(\fra{1}{2}q_1^2(F),q_{1}^2(F)+q_{3}^2(F),q_{1}^2(F)+q_{2}^2(F))]\not\propto (\frac{1}{2}q_1^2(P_1),q_{1}^2(P_1)+q_{3}^2(P_1),q_{1}^2(P_1)+q_{2}^2(P_1)),
\] 
given that $q_{i}^2(F)$ provides a \emph{biased} estimator of $q_{i}^2(P_1)$ under multinomial sampling.
We can, however, remedy this situation by computing unbiased estimators of the squares of the squangles.  
We denote these polynomials as $S_i$, defined through the condition 
\[
E[S_i(F)]=q_i^2(P_1).
\]
Then we redefine our measure to be 
\[
\Delta(F):=\left(\fra{1}{2}S_1(F),S_1(F)+S_3(F),S_1(F)+S_2(F)\right)
\]
and it follows that
\[
E[\Delta(F')]=\det(g)^2E[\Delta(F)],
\]
as required by Property II (strong).

We now discuss how to explicitly compute the unbiased forms $S_i$.
To simplify the presentation, we will denote the probabilities of \emph{distinct} patterns $ijkl$ using the symbols $x_1,x_2,x_3\ldots$ and the corresponding site pattern counts $f_{ijkl}$ using the symbols $X_1,X_2,X_3\ldots$. 
The moment generating function for the multinomial distribution is then expressed as 
\[
f(s_1,s_2,s_3,\ldots ):=E[e^{s_1X_1+s_2X_2+s_3X_3+\ldots }]=(x_1e^{s_1}+x_2e^{s_2}+x_3e^{s_3}+\ldots )^N,
\]
so the expectation value of a monomial in the site pattern counts can then be computed via
\beqn
\label{eq:mogen}
E[X_{1}^{n_1}X_{2}^{n_2}X_{3}^{n_3}\ldots ]=\left.\frac{\partial^{n_1}}{\partial s_1^{n_1}}\frac{\partial^{n_2}}{\partial s_2^{n_2}}\frac{\partial^{n_3}}{\partial s_3^{n_3}}\ldots f(s_1,s_2,\ldots)\right|_{s_1\!=\!s_2\!=\!s_3\!=\!\ldots\!=\!0}.
\eqn

As was alluded to in the proof of Theorem~\ref{thm:sqexp}, complications arise when considering the expectation values of polynomials in the counts $f_{ijkl}$.
Even though the squangles are square free (as the explicit form given in Online Resource 1 shows), when we compute residuals according to Table~\ref{tab:residuals}, the relevant polynomials $q_i^2(F)$ will no longer be square free.
However, we can at least say that each degree six monomial term in $q_i^2(F)$ has no exponents of higher order than a square.
Thus we need only consider the monimals of the form $X_1^2X_2^2X_3^2$, $X_1^2X_2^2X_3X_4$, $X_1^2X_2X_3X_4X_5$, and $X_1X_2X_3X_4X_5X_6$.
Using (\ref{eq:mogen}), we found that, in each case, we can obtain bias-corrected forms by the simple replacement $X_i^2\rightarrow X_i^2-X_i$.
Indeed, following this through one finds:
\beqn
E[(X_1^2-X_1)(X_2^2-X_2)(X_3^2-X_3)]&=N(N-1)\ldots (N-5)x_1^2x_2^2x^2_3,\nonumber\\ E[(X_1^2-X_1)(X_2^2-X_2)X_3X_4]&=N(N-1)\ldots (N-5)x_1^2x_2^2x_3x_4,\\
E[(X_1^2-X_1)X_2X_3X_4X_5]&=N(N-1)\ldots (N-5)x_1^2x_2x_3x_4x_5,\\
E[X_1X_2X_3X_4X_5X_6]&=N(N-1)\ldots (N-5)x_1x_2x_3x_4x_5x_6.
\eqn
We then applied this process to each monomial in $q_i^2(F)$ and divided by the combinatorial factor $N(N-1)(N-2)\ldots (N-5)$ to produce $S_i(F)$; inhomogeneous polynomials with the defining property $E[S_i(F)]=q_i^2(P)$, as required.

Our computations showed that the expansion of each $q_i^2$ has 4008 monomial terms whereas the $S_i$ each have 6688 terms.
This is a significant computational complication as we can efficiently compute $q_i^2(F)$ by simply taking the square $(q_i(F))^2$ (where we recall each $q_i$ only has 96 monomial terms).
However, having computed the explicit polynomial form of each $S_i$ once (we did so in Mathematica~\cite{mathematica}), there is no need to do so again, and having done so repeated numerical evaluation is no great computational obstruction.
We have included the explicit polynomial form of the $S_i$ in the Online Resource 2.

With the unbiased forms $S_i$ in hand, we found that the best performing quartet inference method obtainable is as described by the pseudocode in Table~\ref{tab:unbiasedSquangles}.
We close this section with the conclusion:

\begin{thm}
\label{thm:prop123}
The quartet inference measure and decision rule described in Table~\ref{tab:unbiasedSquangles} satisfies both Property I, Property II (strong), and Property III (see Table~\ref{tab:properties} for definitions).
\end{thm}
\begin{proof}
The result follows from $E[S_i(F)]=q_i^2(P)$ and Theorems~~\ref{thm:sqleaf},  \ref{thm:squanglesleaf}, \ref{thm:thesquangle}, and \ref{thm:sqexp}.
\end{proof}

\begin{table}[tb]
\caption{Our proposed optimal decision rule for using the squangles to infer quartet trees from binary sequence data.}
\centering
\fcolorbox{black}{lightgray}{
\small{
\begin{tabular}{llll}
    \textbf{Input} Four aligned binary sequences;\\
    Compute the site pattern count tensor $F=(f_{i_1i_2i_3i_4})$;\\
	Compute the squangles $(q_1(F),q_2(F),q_3(F))$;\\
	\textbf{If} Ordering perfectly matches that implied by quartet $T_i$ (top three rows of Table~\ref{tab:residuals}); \\			\qquad	\textbf{Return} $T_i$;\\
    \textbf{Else} \\
	\qquad Compute $\hat{u}_1=\frac{1}{2}(q_3(F)-q_2(F))$, $\hat{u}_2=\frac{1}{2}(q_1(F)-q_3(F))$, and $\hat{u}_3=\frac{1}{2}(q_2(F)-q_1(F))$;	 \\ 
     \qquad For $j,k$ such that $\hat{u}_j,\hat{u}_k\geq 0$, compute the bias corrected residuals $S_j(F)$ and $S_k(F)$.\\
	\qquad \textbf{If} $S_j(F)< S_k(F)$; \\
 \qquad \qquad \textbf{Return} $T_j$;\\
 \qquad \textbf{Else} \\
 \qquad \qquad \textbf{If} $S_k(F)< S_j(F)$; \\
 \qquad \qquad \qquad \textbf{Return} $T_k$;\\
 \qquad  \qquad \textbf{Else} \\
 \qquad \qquad \qquad \textbf{Return} tie $T_j$ and $T_k$.\\
\end{tabular}
}
}
\label{tab:unbiasedSquangles}
\end{table}


\section{The edge identities}\label{sec:edgeIdentities}
In this section we discuss the behaviour of the edge identities (defined in Section~\ref{sec:flattenings}) in terms of Properties I, II and III. 

\subsection{Property I}
\label{subsec:edgeprop1}
As we saw in Section~\ref{sec:flattenings}, in the context of quartet trees and binary sequence data, the edge identities are the cubic minors of the three flattenings $\text{Flat}_1(P)$, $\text{Flat}_2(P)$ and $\text{Flat}_3(P)$. 
For the purpose of this discussion, we denote the $(i,j)$ cubic minor of the flattening $\text{Flat}_1(P)$ as $m_{ij}(\text{Flat}_1(P))$, or simply as $m_{ij}$.

We begin our discussion of Property I for the edge identities with a focus on the action of the stabilizer subgroup, $\text{Stab}(T_1)$, of the quartet $T_1\!=\!12|34$.  This includes uncovering the exact representation of $\text{Stab}(T_1)$ acting on the minors.  


From the results of Section~\ref{sec:taxonperms}, we find that $\text{Stab}(T_1)$ acts on the set of 16 cubic minors by signed permutations.
Specifically, for all $i,j=1,2,3,4$:
\[
\begin{array}{ll}
m_{1j}(K\text{Flat}_1(P))=-m_{1j}(\text{Flat}_1(P));&
\phantom{++}m_{2j}(K\text{Flat}_1(P))=m_{3j}(\text{Flat}_1(P));\\
m_{3j}(K\text{Flat}_1(P))=m_{2j}(\text{Flat}_1(P));&
\phantom{++}m_{4j}(K\text{Flat}_1(P))=-m_{4j}(\text{Flat}_1(P));
\end{array}
\]
and
\[
m_{ij}(\text{Flat}_1(P)^T)=m_{ji}(\text{Flat}_1(P)).
\]
Thus we see that the minors break up into six (signed) orbits under this action:
\beqn
\{m_{11}\}, \{m_{12},\!m_{21},\!m_{31},\!m_{13}\}, \{m_{14},\! m_{41}\}, \{m_{22},\! m_{32}, \!m_{23}, \!m_{33}\},\{m_{24},\! m_{34}, \!m_{42}, \!m_{43}\}, \{m_{44}\}.\nonumber
\eqn

Taking a multinomial sample $F\sim \text{MultiNom}(P,N)$ and fixing an orbit, we see that the sum of squares $\sum_{m_{ij}\in \text{orbit}}|m_{ij}(F)|^2$ is explicitly invariant under the action of $\text{Stab}(T_1)$.  
If we define $\Delta_1$ to be this sum and analogously define $\Delta_2$ and $\Delta_3$, then $\Delta(F) := (\Delta_1, \Delta_2,\Delta_3)$ is a quartet measure that satisfies Property I.  

However, we could have instead used certain linear combinations of the minors from each orbit and still produce a polynomial invariant under $\text{Stab}(T_1)$ and thus the analogous polynomials together form a Property I satisfying quartet measure.  
For example, we could define $\Delta_1$ to be $|m_{11}(F)+m_{44}(F)|^2+|m_{11}(F)-m_{44}(F)|^2$.  
In general, it is possible to take linear combinations of the minors which transform as one-dimensional \emph{linear} representations of the stabilizer $\text{Stab}(T_1)$ and then take sums of squares of thereof as an inference measure (this can be done systematically using the methods discussed in \cite{sumner2009}).
We performed this analysis but omit the details here, since we found that there was no particular gain in statistical power over the straightforward sum of squares described above.

The theoretical conclusion is that leaf permutation symmetries alone are not enough to uniquely determine a choice of measure constructed from the minors. 
Further, there is no reason at all to expect this measure will satisfy Property II in the weak or strong form. 
This discussion does however raise the natural question of whether there perhaps exists a linear combination of minors which satisfies both Property I and II.
Of course, this linear combination is exactly the (binary) squangle constructed in the previous section.  
Importantly, this fact is specific to this binary case, and in Section~\ref{sec:generalize} and the appendix we discuss why this is a special feature restricted to the binary Markov model.  


\subsection{Signs for the edge identities}
\label{sec:edgesigns}

We now turn to exploring Property III in relation to the edge identities.
The results presented in this section will only be valid under a continuous-time Markov process.

Explicit computation shows that, as a polynomial in the variables $P=(p_{ijkl})_{i,j,k,l\in \{0,1\}}$, each minor of the flattening $\text{Flat}_1(P)$ can be expressed as a linear combination of a minor of  $\text{Flat}_2(P)$ with a minor of  $\text{Flat}_3(P)$.
We present these relationships in Table~\ref{tab:signedminors}.
We stress that these are algebraic relationships between the minors as polynomials in the variables $(p_{ijkl})$, valid for all tensors $P$.

If we fix a tensor $P_1$ arising from $T_1$, we see that we may re-express the vanishing of a given minor of $\text{Flat}_1(P_1)$ as an equality between the corresponding minors in $\text{Flat}_2(P_1)$ and $\text{Flat}_3(P_1)$.
Importantly, it then turns out that (under mild conditions discussed shortly) there exists a \emph{positive} parameter $u$ (whose exact value depends on the specific choice of the model parameters defining $P_1$) such that the value of the respective minors is either $+u$ or $-u$.
The relevant signs are also provided in Table~\ref{tab:signedminors}.

Taking this sign information into account leads to an important modification of the quartet inference measure obtained using the edge identities.
This is implemented in the least squares framework with residual sums of squares exactly analogous to the signed squangles described in Table~\ref{tab:residuals}.  
Presently, we describe the conditions that lead to this additional sign information.

The (mild) condition we impose is that the $2\times 2$ Markov matrices $M$ on the leaves of the quartet tree have positive determinant: $\det(M)>0$.
The reader should note that this is a biologically reasonable condition where evolutionary times are generally of the order where a probability of substitution is smaller than the probability of no substitution.
This is also the case when we consider a continuous-time implementation of the underlying Markov process, so $M=e^{Qt}$ for some rate matrix $Q$ and hence $\det(M)=e^{\text{tr}(Qt)}>0$.

Assuming this condition, the inverses of such Markov matrices have entries with signs given by:
\[
M^{-1}=
\left(
\begin{matrix}
+ & - \\
- & +
\end{matrix}
\right).
\]
Hence, if we take a Kronecker product of two such matrices we have the signed form:
\[
M^{-1}\otimes M^{-1}=
\left(
\begin{matrix}
+ & - & - & + \\
- & + & + & - \\
- & + & + & - \\
+ & - & - & +
\end{matrix}
\right).
\]

Arguing as we did in Section~\ref{subsec:construct}, taking a clipped tensor $\widetilde{P}_1$ arising on $T_1$, we have
\[
\text{Flat}_2(\widetilde{P}_1)=\text{Flat}_3(\widetilde{P}_1)=
\left(
\begin{matrix}
x & 0 & 0 & 0 \\
0 & y & 0 & 0 \\
0 & 0 & z & 0 \\
0 & 0 & 0 & w
\end{matrix}
\right),
\]
where $x,y,z,w>0$.
Organizing the minors of these flattenings into the corresponding cofactor matrices (that is, the $4\times 4$ matrix where the $(i,j)$ entry is $(-1)^{i+j}$ times the $(i,j)$ minor), we have
\[
\text{Cof}(\text{Flat}_2(\widetilde{P}_1))=
\text{Cof}(\text{Flat}_3(\widetilde{P}_1))=
\left(
\begin{matrix}
yzw & 0 & 0 & 0 \\
0 & xzw & 0 & 0 \\
0 & 0 & xyw & 0 \\
0 & 0 & 0 & xyz
\end{matrix}
\right)
.
\]

Recalling that the cofactor matrix can be expressed as 
\[
\text{Cof}(A)=\det(A){A^{-1}}^{T},
\]
it follows that the cofactor matrix is multiplicative: $\text{Cof}(AB)=\text{Cof}(A)\text{Cof}(B)$.
From the above expressions we may conclude that the cofactor matrices of the flattenings have, for any $P_1=M_1\otimes M_2\otimes M_3\otimes M_4\cdot \widetilde{P}_1$ arising on quartet $T_1$, the signed form:
\[
\text{Cof}(\text{Flat}_2(P_1))=\text{Cof}(M_1\otimes M_3)\cdot \text{Cof}(\text{Flat}_2(\widetilde{P}_1))\cdot \text{Cof}((M_2\otimes M_4)^T)=
\left(
\begin{matrix}
+ & - & - & + \\
- & + & + & - \\
- & + & + & - \\
+ & - & - & +
\end{matrix}
\right),
\]
and, similarly, the same result holds for the signed form of $\text{Cof}(\text{Flat}_3(P_1))$.

We have used this result to produce the sign information given in Table~\ref{tab:signedminors}. 
The information in this table can be used to produce a measure that we refer to as the ``signed minor'', with the specific algorithm described in Table~\ref{tab:minorsMethod}. 

\begin{table}[h]
\caption{Algebraic relationships between the minors of the three flattenings.
Each row is interpreted as specifying the choice of algebraic relationship $a\!+\! (b\!-\!c)\!=\!0$ or $a\!-\! (b\!-\!c)\!=\!0$ for the three minors $a,b,c$ indicated in the row evaluated on the flattenings $\text{Flat}_1$, $\text{Flat}_2$, and $\text{Flat}_3$, respectively.
Additionally, for each row, the fifth column indicates that, on any tree tensor $P_1$ arising from tree $T_1$,  the non-zero minors $b$ and $c$ satisfy $b\!=\!c\!=\!\pm u$ for some positive parameter $u\!\equiv\! u(P_1)>0$ (under the condition that the $2\times 2$ Markov matrices on the leaves of the tree have positive determinant).
For example, the fourth row indicates that: (i) for all tensors $P\in U$, the $(2,2)$ minor of $\text{Flat}_1(P)$, plus the $(1,4)$ minor of $\text{Flat}_2(P)$, minus the $(2,3)$ minor of $\text{Flat}_3(P)$ is equal to zero, and (ii) for all tree tensors $P_1$ arising from $T_1$, the $(1,4)$ minor of $\text{Flat}_2(P)$ is negative and equal to the $(2,3)$ minor of $\text{Flat}_3(P)$.}
\centering
\begin{tabular}{c|c|c|c|c}
$\text{Flat}_1$ & $\text{Flat}_2$ & $\text{Flat}_3$ & $a\pm(b-c)=0$ & $\pm u$\\
\hline
$(1,1)$ & $(1,1)$ & $(1,1)$ & $-$ & $+$\\
$(1,2)$ & $(1,2)$ & $(2,1)$ & $-$ & $+$\\
$(2,1)$ & $(1,3)$ & $(1,3)$ & $+$ & $-$\\
$(2,2)$ & $(1,4)$ & $(2,3)$ & $+$ & $-$\\
$(1,3)$ & $(2,1)$ & $(1,2)$ & $+$ & $+$\\
$(1,4)$ & $(2,2)$ & $(2,2)$ & $+$ & $+$\\
$(2,3)$ & $(2,3)$ & $(1,4)$ & $-$ & $-$\\
$(2,4)$ & $(2,4)$ & $(2,4)$ & $-$ & $-$\\
$(3,1)$ & $(3,1)$ & $(3,1)$ & $-$ & $-$\\
$(3,2)$ & $(3,2)$ & $(4,1)$ & $-$ & $-$\\
$(4,1)$ & $(3,3)$ & $(3,3)$ & $+$ & $+$\\
$(4,2)$ & $(3,4)$ & $(4,3)$ & $+$ & $+$\\
$(3,3)$ & $(4,1)$ & $(3,2)$ & $+$ & $-$\\
$(3,4)$ & $(4,2)$ & $(4,2)$ & $+$ & $-$\\
$(4,3)$ & $(4,3)$ & $(3,4)$ & $-$ & $+$\\
$(4,4)$ & $(4,4)$ & $(4,4)$ & $-$ & $+$
\end{tabular}
\label{tab:signedminors}
\end{table}

\begin{table}[h]
\caption{The ``signed minor'' method for computing the measure $\Delta_1$ for tree $T_1$. The measures $\Delta_2$ and $\Delta_3$ are computed similarly.}
\centering
\fcolorbox{black}{lightgray}{
\small{
\begin{tabular}{llll}
    \textbf{Input} Four aligned binary sequences;\\
    Compute the site pattern tensor $F$;\\
	Compute the flattenings $\text{Flat}_1(F)$, $\text{Flat}_2(F)$, $\text{Flat}_3(F)$;\\
    Set $\text{RSS} = 0$;\\
	\textbf{For each} row $(a, b, c, \pm)$ in Table \ref{tab:signedminors}; \\			
    \qquad	$\hat{u} = (b+c)/2$;\\
    \qquad \textbf{If} $\pm \hat{u} > 0$ (i.e. $\hat{u}$ has the expected sign);\\
    \qquad \qquad $\text{RSS}  \leftarrow \text{RSS}  + 0.5a^2$.\\
    \qquad \textbf{Else} \\
	\qquad \qquad $\text{RSS}  \leftarrow \text{RSS}  + b^2 + c^2$.\\
	$\Delta_1 = \text{RSS} $\\

\end{tabular}
}
}
\label{tab:minorsMethod}
\end{table}


\section{Simulation study}
\label{sec:simulations}

\subsection{Phylogenetic methods tested}

We conducted a comprehensive simulation study to compare the accuracy of the inference methods described in this paper. 
To facilitate the discussion we use the following abbreviations:
\begin{itemize}
	\item The measure formed from the binary squangles without the residual sum of squares decision rule (Section \ref{sec:quartetInference}). 
	In other words, compute the squangles $q_1$, $q_2$, and $q_3$ and return the tree with the squangle that is closest to zero: ``unsigned squangles'' or US;
	\item The binary squangles using the residual sum of squares decision rule given in Table~\ref{tab:residuals}: ``signed squangles'' or SS;
	\item The binary squangles with the residual sum of squares decision rule and corrected for bias in the estimators, as described in Table~\ref{tab:unbiasedSquangles}: ``bias-corrected signed squangles'' or CSS;
	\item The measure $\Delta_k = \Sigma_{i,j=1,2,3,4} m_{ij}^2(\text{Flat}_k(F))$ formed from the sum of squares of the 16 matrix minors (the edge identities) without the residual sum of squares decision rule: ``unsigned minors'' or UM;
	\item The measure formed from the 16 matrix minors (the edge identities) with the residual sum of squares decision rule, as described in Table~\ref{tab:minorsMethod}: ``signed minors'' or SM;
	\item Neighbor-joining on distances that have been corrected for multiple substitutions using the formula $d_{cor} = -0.5 \log(1 - 2d_{obs})$, where $d_{obs}$ is the proportion of sites that differ between two aligned sequences: ``neighbor-joining'' or NJ.
	\item Neighbor-joining on distances that have not been corrected for multiple substitutions: ``uncorrected neighbor-joining'' or UNJ.
	\item The method proposed by \cite{eriksson2008} that is based on singular value decomposition of tensor flattenings: ``ErikSVD''. 
	\item The subsequent modification to the ErikSVD method proposed by \cite{fernandez2015} that  normalises  the tensor flattening before applying singular value decomposition: ``Erik+2''.
\end{itemize}

Where known, the statistical properties of these methods are summarised in Table~\ref{tab:propsofmethods}.

\begin{table}[tb]
	\caption{Properties of the quartet inference methods discussed in this work (see Table~\ref{tab:properties} for definitions).
	For each property that is satisfied, the final column provides a reference to the relevant result and/or discussion.}
	\centering
	\begin{tabular}{c|ccccc}
		Method & Prop. I &  Prop. II (weak) & Prop. II (strong) & Prop. III & References (resp.) \\
		\hline
		US & \cmark & \cmark & \xmark & \xmark & Thms~\ref{thm:sqleaf},  \ref{thm:squanglesleaf}, \ref{thm:thesquangle} \\
		SS & \cmark & \cmark & \xmark & \cmark &  Thms~\ref{thm:sqleaf},  \ref{thm:squanglesleaf}, \ref{thm:thesquangle}, \ref{thm:sqexp} \\
		CSS & \cmark & \cmark & \cmark & \cmark & Thms~\ref{thm:sqleaf},  \ref{thm:squanglesleaf}, \ref{thm:thesquangle}, \ref{thm:sqexp}, \ref{thm:prop123}  \\
		UM & \cmark & \xmark & \xmark & \xmark & Sec~\ref{subsec:edgeprop1} \\
		SM & \cmark & \xmark & \xmark & \cmark & Sec~\ref{subsec:edgeprop1}, Tab~\ref{tab:signedminors}\\
		NJ & \cmark  & \cmark & \cmark & \cmark & Thm~\ref{thm:nj}
	\end{tabular}
	\label{tab:propsofmethods}
\end{table}

\subsection{Generation of simulated data}

All the simulations use a continuous time, symmetric Markov model.  
Edge length parameters correspond to the probability of a change along an edge (as opposed to the expected number of changes). Data were simulated on one of four types of tree: ``Felsenstein'', ``Farris'', ``balanced'', or ``unbalanced star'' (Figure~\ref{fig:felsfarris}). 
These trees were chosen as they have been widely studied in the literature concerning accuracy of different phylogenetic methods \cite{felsenstein1978,huelsenbeck1993success,swofford2001bias}. 
Many methods are known to have biases on these tree shapes either (i) negatively, towards getting an incorrect tree (for the Felsenstein shape), or (ii) positively, towards getting the correct tree (for the Farris shape). 

\tikzstyle{inner} = [ draw, circle, inner sep=0pt, minimum size=1pt, fill]
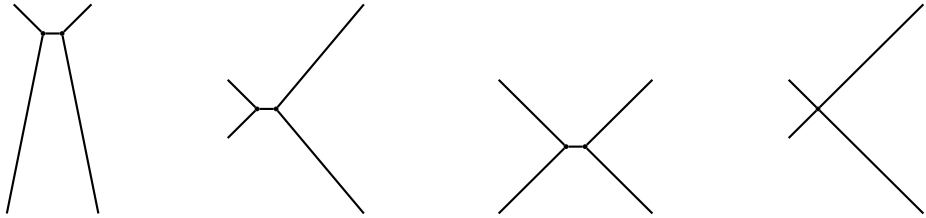
\begin{figure}[bt]
	\caption{Left to right: ``Felsenstein tree'', ``Farris tree'', ``balanced tree'', and ``unbalanced star''.}
	\centering
	\begin{tabular}{ccccccc}
		\begin{tikzpicture}[scale=.5,thick]
		\node(1) at (0,1) {};
		\node(2) at (0,-5) {};
		\node[inner](a) at (1,0) {};
		\node[inner](b) at (1.5,0) {};
		\node(3) at (2.5,1) {};
		\node(4) at (2.5,-5) {};
		\draw[-] (1) -- (a) -- (b) -- (4);
		\draw (2) -- (a);
		\draw (3) -- (b);
		\end{tikzpicture}
		
		&\phantom{how}&
		
		\begin{tikzpicture}[scale=.5,thick]
		\node(1) at (0,1) {};
		\node(2) at (0,-1) {};
		\node[inner](a) at (1,0) {};
		\node[inner](b) at (1.5,0) {};
		\node(3) at (4,3) {};
		\node(4) at (4,-3) {};
		\draw[-] (1) -- (a) -- (b) -- (4);
		\draw (2) -- (a);
		\draw (3) -- (b);
		\end{tikzpicture}  
		
		&\phantom{how}&
		
		\begin{tikzpicture}[scale=.5,thick]
		\node(1) at (0,2) {};
		\node(2) at (0,-2) {};
		\node[inner](a) at (2,0) {};
		\node[inner](b) at (2.5,0) {};
		\node(3) at (4.5,2) {};
		\node(4) at (4.5,-2) {};
		\draw[-] (1) -- (a) -- (b) -- (4);
		\draw (2) -- (a);
		\draw (3) -- (b);
		\end{tikzpicture} 
		
		&\phantom{how}&
		
		\begin{tikzpicture}[scale=.5,thick]
		\node(1) at (0,1) {};
		\node(2) at (0,-1) {};
		\node[inner](a) at (1,0) {};
		\node(3) at (4,3) {};
		\node(4) at (4,-3) {};
		\draw[-] (1) -- (a) -- (4);
		\draw (2) -- (a);
		\draw (3) -- (a);
		\end{tikzpicture}      
	\end{tabular}
	\label{fig:felsfarris}
\end{figure}

We conducted three different sets of simulations. For all scenarios we simulated 1000 trees for each parameter combination and recorded how many were correctly inferred.

The first set of simulations explored the effect of sequence length on accuracy of the different methods. 
Data were simulated on each type of tree.  Long branches had a 0.3 probability of a change and short branches had a 0.05 probability of a change. The sequence length was varied from 50 to 1600 in steps of 50. 

The second set of simulations explored the effect of internal branch length on accuracy of the different methods. 
Data were simulated on each type of tree excluding the star tree. 
Long pendant branches had a 0.3 probability of a change and short pendant branches had a 0.05 probability of a change. 
The internal branch length was varied from 0 to 0.1 in steps of 0.01. 
Sequence length was fixed at 800 characters. 

The third set of simulations solely focused on the ``Felsenstein'' tree and explored the effect of varying both the length of the two long branch lengths and the length of the three short branches. The short branch length was varied from 0.01 to 0.1 in steps of 0.01. The long branch length was varied from 0.1 to 0.4 in steps of 0.03. Sequence length was fixed at 400.

\subsection{Simulation results}

We present results for a subset of the methods and generating trees. 
The full simulation results, including heat maps, are available in Online Resource 3.

The results of the first set of simulations on the ``Felsenstein'' tree are presented in Figure ~\ref{fig:SL_Fel}.  The unsigned variants both perform poorly, while SM, SS, and NJ perform roughly equally well and CSS, the bias-corrected squangles, is the most accurate. 
The results on the ``Farris'' tree are presented in Figure~\ref{fig:SL_Far}. For this tree the unsigned variants are most accurate, particularly for shorter sequence lengths. UM is more than 80\% accurate even for sequence lengths of 50. 
SM, SS, and NJ perform roughly equally well and CSS, the bias-corrected squangles, is the least accurate.  
For the ``balanced tree'' (Figure~\ref{fig:SL_Bal}) all the signed methods and NJ performed about equally and were much more accurate than the unsigned methods and the methods based on singular value decomposition. 

The simulations on the ``unbalanced star'' tree  give us a more explicit opportunity to investigate the effect of (positive) bias inherent in the results for the Farris tree (Figure~\ref{fig:SL_Far}). 
If a method is unbiased it should have no preference for any one quartet (and hence return each quartet roughly 1/3 of the time). 
In order to investigate this, the number of times the tree which groups the two long edges together was recorded (Figure~\ref{fig:SL_UnbalStar}). 
UNJ is by far the most biased method.
This is followed by UM which returns the tree that pairs the two long branches over 80\% of the time, and US which returns this tree about 65\% of the time. 
SM, SS and NJ are all less biased but return the tree that pairs the two long branches about 40\% of the time. 
CSS is the only method tested that appears to be unbiased.
  
Overall the simulation scenarios, performance of our binary (two-state) implementations of ErikSVD and Erik+2 was relatively poor. 
This is in contrast to the excellent performance of this approach for four-state data reported in \cite{fernandez2015}. 
It is not immediately obvious why these methods are less effective on binary sequences.

\begin{figure}[ht]
\caption{Accuracy of nine different phylogenetic methods for data simulated under varying sequence lengths on the ``Felsenstein'' tree (short branch lengths 0.05 and long branch lengths 0.3).
The performance of NJ and SS is almost indistinguishable.}
\centering
\includegraphics[width = 1.0\textwidth]{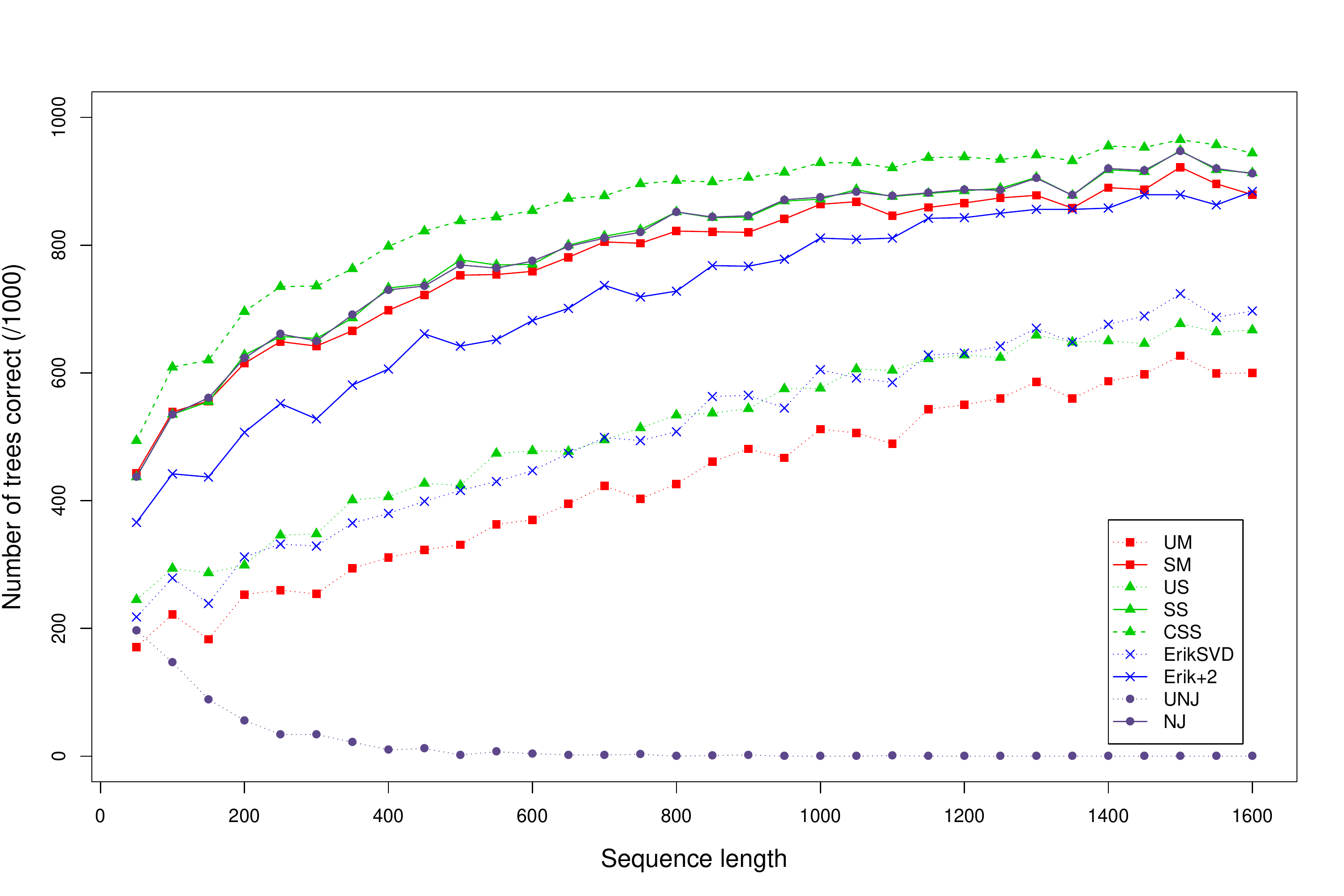}
\label{fig:SL_Fel}
\end{figure}

\begin{figure}[ht]
	\caption{Accuracy of nine different phylogenetic methods for data simulated under varying sequence lengths on the ``Farris'' tree (short branch lengths 0.05 and long branch lengths 0.3). 
Note that performance of NJ and SS is almost indistinguishable.
The high accuracy for some methods reflects bias towards inferring the correct tree (c.f. the results for the ``unbalanced star'' in Fig~\ref{fig:SL_UnbalStar}).}
\centering	
\includegraphics[width =  1.0\textwidth]{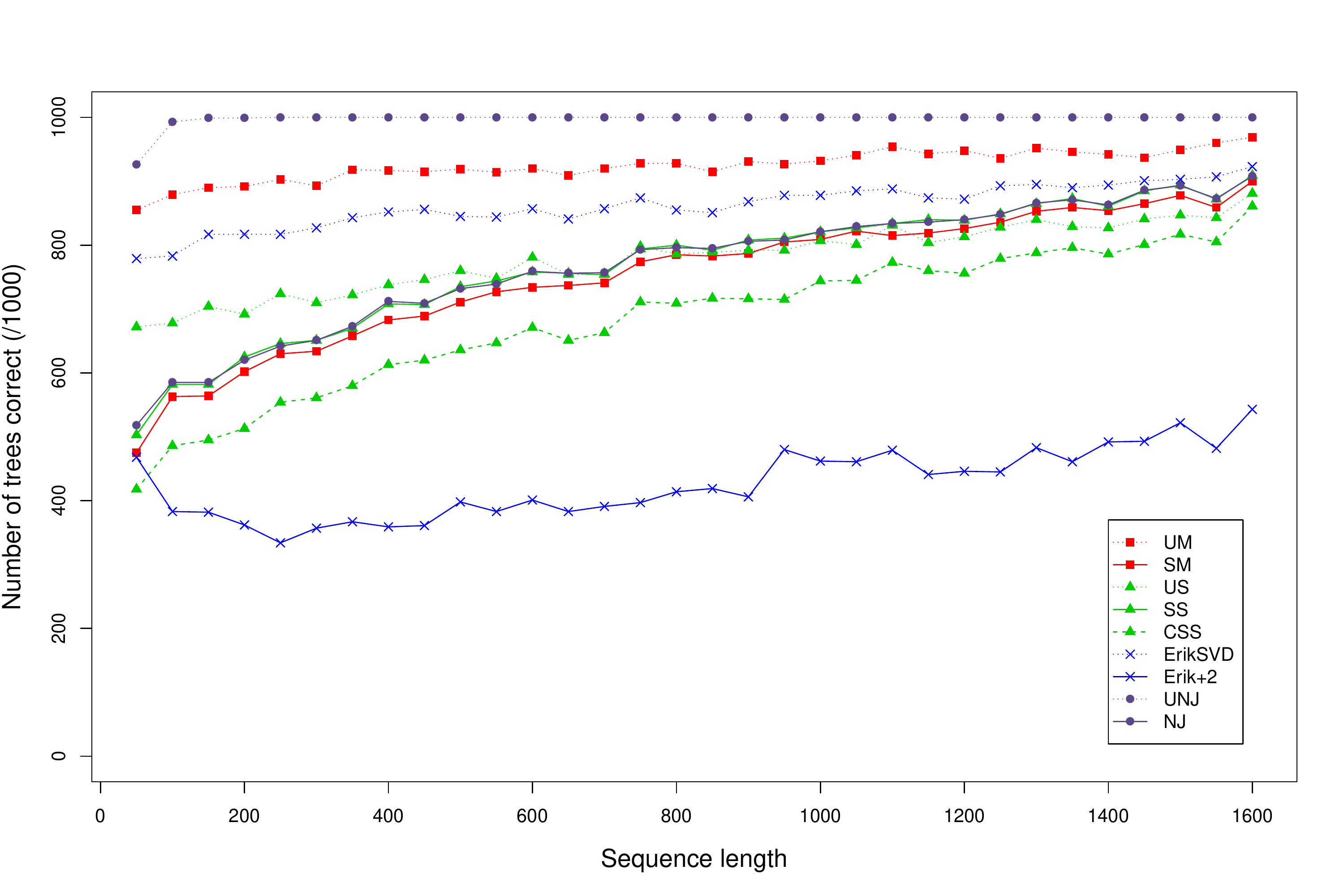}
\label{fig:SL_Far}
\end{figure}

\begin{figure}[ht]
	\caption{Accuracy of nine different phylogenetic methods for data simulated under varying sequence lengths on the balanced tree (internal branch length 0.05 and pendant branch lengths 0.3). Performance of the methods SM, CSS, SS, NJ and UNJ is almost indistinguishable and better than the performance of UM, US, ErikSVD and Erik+2.}
\centering	
\includegraphics[width =  1.0\textwidth]{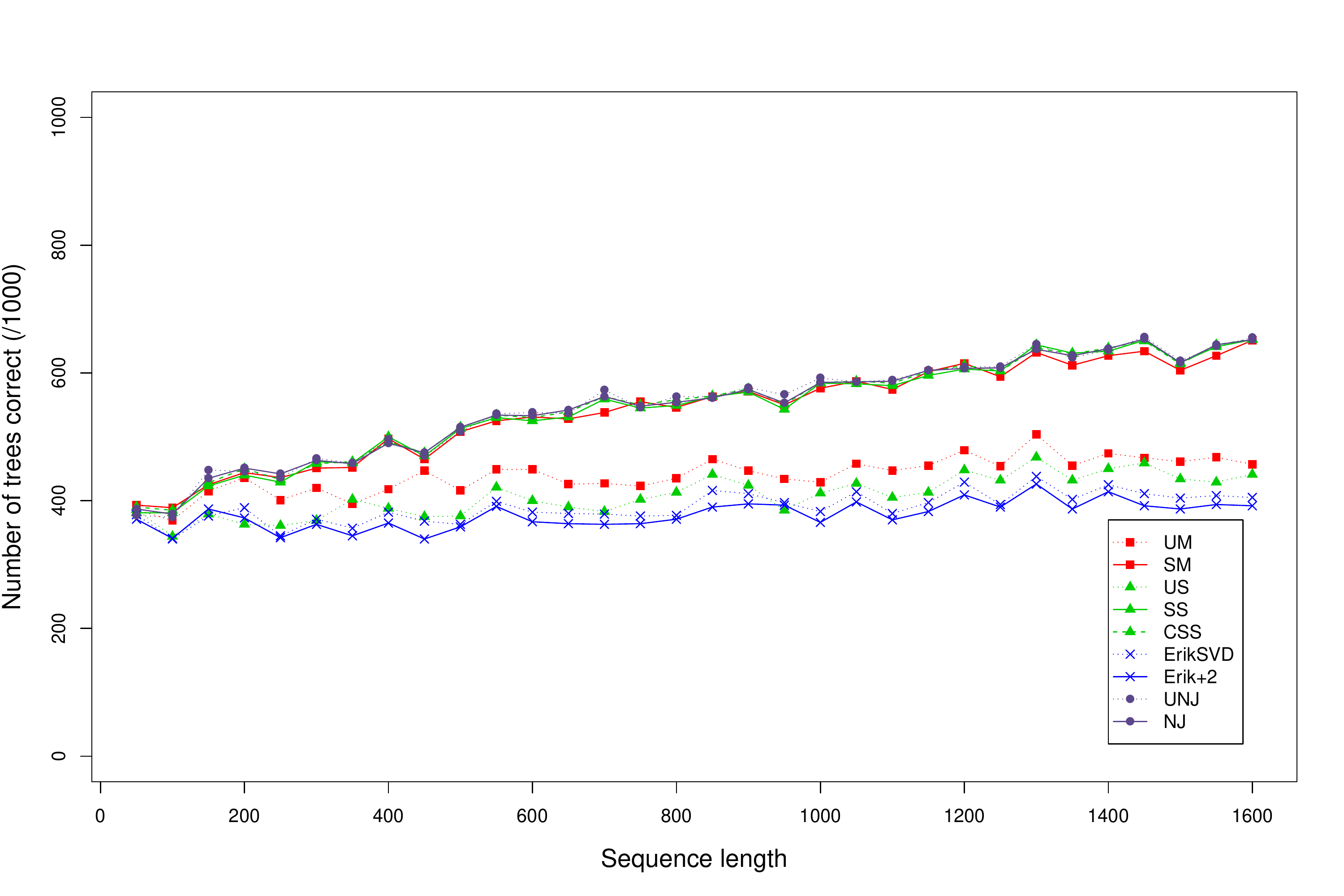}
\label{fig:SL_Bal}
\end{figure}

\begin{figure}[ht]
	\caption{	
	Performance of nine different phylogenetic methods for data simulated under varying sequence lengths on
the unbalanced star tree (short branch lengths 0.05 and long branch lengths 0.3). The dashed horizontal line at 333.3 indicates ideal performance of an unbiased method.
}
\centering	
\includegraphics[width = 1.0\textwidth]{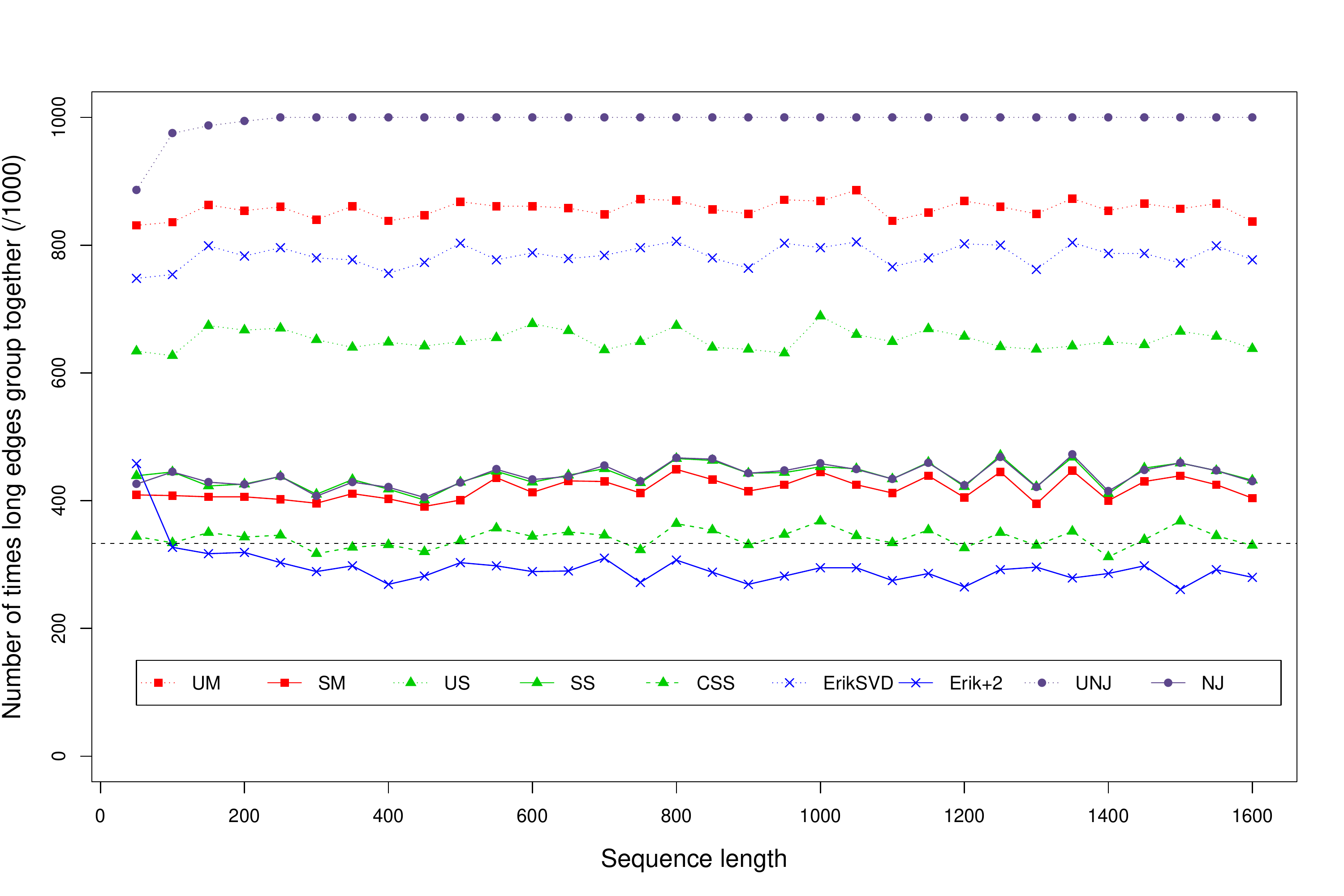}
\label{fig:SL_UnbalStar}
\end{figure}

The results of the second set of simulations on the ``Felsenstein'' tree are presented in Figure~\ref{fig:EL}. The unsigned variants both perform relatively poorly, SM, SS and NJ perform roughly equally well, and the CSS is the most accurate for all internal branch lengths tested.

\begin{figure}[ht]
\caption{Accuracy of nine different phylogenetic methods for varying internal edge lengths on the ``Felsenstein'' tree (short pendent branch lengths 0.05 and long pendent branch lengths 0.3). The internal branch length was varied from 0 to 0.1 in steps of 0.01. 
Sequence length was fixed at 800 characters. }
\centering
\includegraphics[width =  1.0\textwidth]{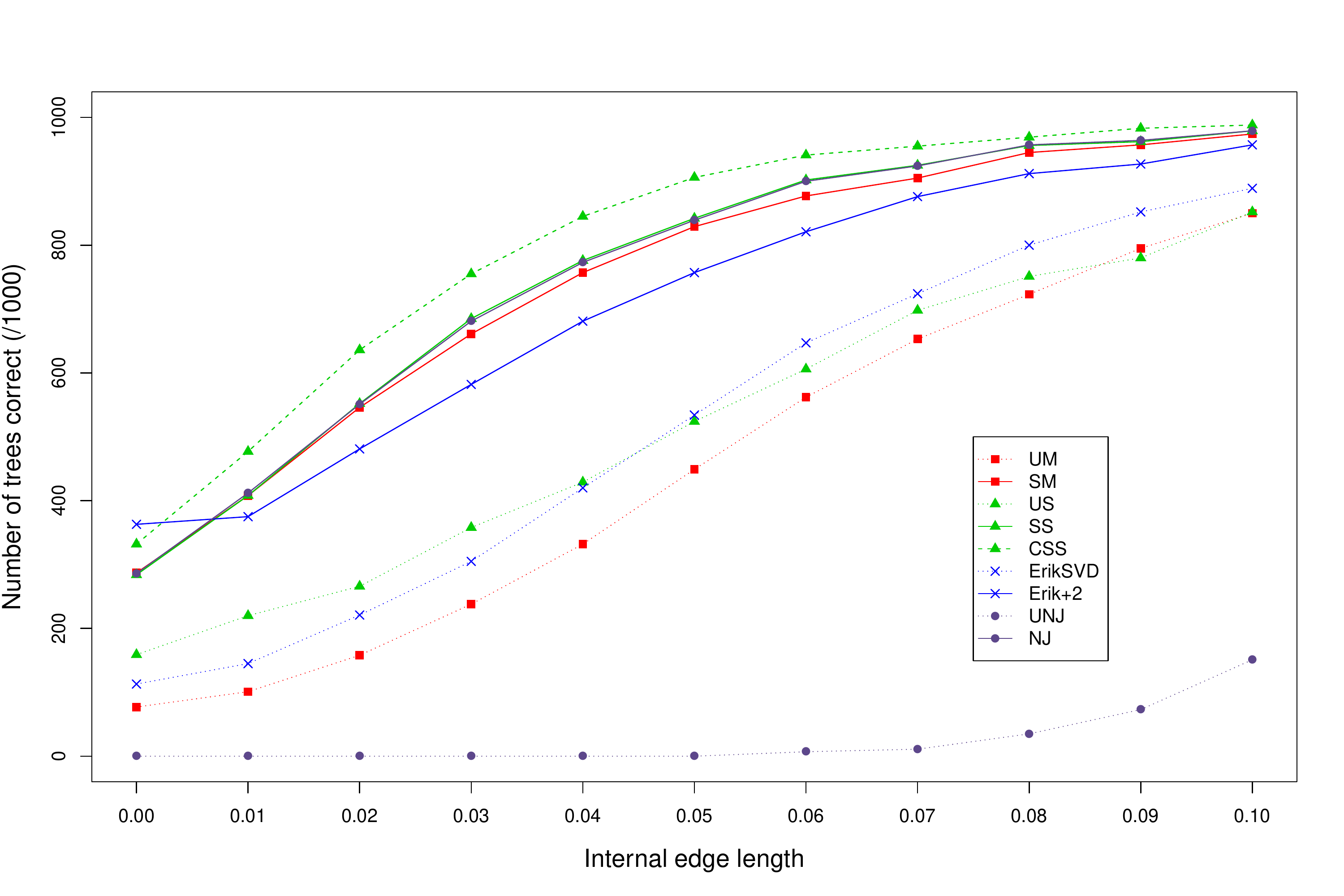}
\label{fig:EL}
\end{figure}
 
The results of the third set of simulations are presented as a series of heat maps in the Online Resource 3. 
Averaged over all the combinations of short and long branch lengths tested the methods ranked as follows for accuracy: CSS (81.4\%), NJ (77.4\%), SS (77.4\%),  SM (76.7\%), Erik+2 (71.8\%), US (56.1\%), ErikSVD (56.0\%), UM (52.1\%), UNJ (40.3\%).

Overall our results indicate that using a method (CSS) which has Properties I, II and III provides a highly accurate and unbiased method of quartet topology inference. The moderately improved performance of CSS relative to SS shows the benefit of correcting for the bias in the squangles, but that we can still reasonably use the squangles without this correction.  For more on the future of bias correcting methods based on Markov invariants beyond binary state models, see Section~\ref{sec:generalize}. 
%

\section{Discussion and future work}
\label{sec:generalize}

The above analysis focuses exclusively on the binary case $k\!=\!2$.
However, biologists are usually interested in studies where $k\!=\!4$, the DNA case.
Here we discuss the extension of the above results to $k\!=\!4$. 

\subsection{The minors} 
In the Appendix (Online Resource 4) we establish that, when $k\!=\!2$, the 48 minors (16 minors from each of the 3 flattenings) form a 32-dimensional invariant subspace under the action of $\text{GL}(4)\times \text{GL}(4)$ (expressed as left and right matrix multiplication). 
Further, in this scenario the binary squangles are elements of this invariant subspace and thus occur as certain linear combinations of the minors.  

For the DNA case of $k\!=\!4$, similar representation theoretic arguments (as given in the Appendix, Online Resource 4) establish that the invariant subspace formed from the minors of the flattenings (now degree 5 minors of $16\times 16$ matrices) \emph{does not contain any Markov invariants}.  
This happens because the rank conditions on flattenings are invariant under the action of $\text{GL}(16)\times \text{GL}(16)$ (again expressed as left and right matrix multiplication), whereas the Markov invariants are valid only under the two-step subgroup restriction: 
\begin{enumerate}
\item $\text{GL}(16)\times \text{GL}(16)$ to $\times^4\text{GL}(4)\equiv \left( \text{GL}(4)\times \text{GL}(4)\right)\times \left( \text{GL}(4)\times \text{GL}(4)\right)$;
\item each copy $\text{GL}(4)$ thereof to Markov matrices $\mathcal{M}_4$.
\end{enumerate}
For $k \!=\! 2$, it turns out there are so few possible invariant subspaces that the Markov invariants (binary squangles) happen to be in the subspace of polynomials spanned by the minors, i.e. they are linear combinations of the minors. 
For $k\!=\!4$, the minors and DNA squangles lie in distinct $\times^4 \text{GL}(4)$ invariant subspaces and it follows there is no linear combination of minors forming a Markov invariant (see the Appendix, Online Resource 4) and hence no chance a quartet inference measure formed from the minors can be made to satisfy Property II (strong).   

\begin{thm}
\label{thm:notminors}
The DNA squangles $(k\!=\!4)$ do not occur as linear combinations of minors of flattenings.
As a consequence, there is no quartet inference measure based on minors and edge identities satisfying Property II (weak or strong).
\end{thm}
\begin{proof}
See Appendix, Online Resource 4.
\end{proof}

On the other hand, for any $k$ the minors will continue to transform as a signed permutation representation of the relevant stabilizer subgroups so it is no problem to ensure Property I by a suitable choice of measure (any sum of squares of an orbit under the stabilizer subgroup will do).

Additionally, assuming that the relevant Markov matrices are sufficiently close to the identity matrix, similar arguments to those given in Section~\ref{sec:edgesigns} can be made to determine the signs for the edge identities on quartets for any $k$ (we leave the details of this for future work).


\subsection{The squangles} 
As described in~\cite{holland2012}, the theory we presented here to build the basic residual sum of squares rule for $k\!=\! 2$ extends to $k\! =\! 4$ to provide a signed residual sum of squares rule for DNA data. 
However, this derivation does not include the computation of the unbiased forms $S_i$ of the squares of the DNA squangles.
While the representation-theoretic arguments for existence of the Markov invariants are similar, their construction is more complicated and there is no known way to compute them as a minor of a transformed flattening.
We emphasize that the representation theory showing our tree measure has Property II (strong) extends to showing there is such a measure for $k\! =\! 4$, as well as the behaviour under taxon permutations ensuring Property I.  

An important addition to this paper over~\cite{holland2012} is our discussion of unbiased estimators of the parameters involved in the decision rule as discussed in Section~\ref{sec:quartetInference}. 
Since the binary squangles have relatively few terms, computing the unbiased forms of their squares is feasible by explicit squaring and bias correcting term by term.
In particular, the binary squangles are cubic and square-free and thus we know that the only correction we need to make is for squared variables. 
However, for $k\! =\! 4$ the DNA squangles are degree 5 polynomials in 256 variables with 66,744 terms each.
Additionally, and most importantly, their explicit polynomial form is only known in a non-standard basis analogous to the change of basis used at the start of Section~\ref{sec:squangles} (details are given in \cite{sumner2009}).  
Given that these polynomials would need to be squared and transformed to the natural (probability) basis, we consider the development of an unbiased square of the squangles for $k \! =\! 4$ a challenging open problem.  

\vspace{1em}
\noindent\textbf{Open Problem}: Compute unbiased forms for the DNA squangles.
\vspace{1em}

We close by pointing out that it is easy to argue that the SVD approach satisfies Property I but certainly does not satisfy Property II and it is not at all clear how to construct a correspondingly unbiased version of the SVD approach.

\subsection*{Funding}
This work was supported by the Australian Research Council Discovery Early Career Fellowship DE130100423 (JGS) and the University of Tasmania Visiting Scholars Program (AT).

\subsection*{Acknowledgement}
We would like to thank the two anonymous reviewers whose thoughtful and careful reading of our manuscript led to a greatly improved final version.

\bibliographystyle{plain}
\bibliography{master}

\end{document}